\documentclass[aps,prr,reprint,twocolumn,notitlepage,superscriptaddress]{revtex4}
\usepackage{CJK}

\usepackage[dvipdfmx]{graphicx,color}	%画像
\usepackage{amsmath}
\usepackage{amsfonts}
\usepackage{amssymb}	%数式たち
\usepackage{amsthm}  %定理など
  
  \newtheorem*{thm}{Theorem}%[section]
  \newtheorem{prop}{Proposition}[section]
  
  \newtheorem{lem}{Lemma}[section]
\usepackage{bm}
\usepackage{bbm}  %白抜きの1
\usepackage{braket}
\usepackage{mathtools}
\usepackage{stmaryrd}
\usepackage{cases}
\usepackage{color}
\usepackage{leftidx}
\usepackage{physics}
\usepackage{array}  %表立ての機能追加
\usepackage{float}  %図の自動配置
\usepackage{framed}
\usepackage{tikz}  %グラフ描画
\usetikzlibrary{intersections,calc,arrows.meta}  %グラフ描画
\usepackage{comment}
\usepackage[driverfallback=dvipdfmx,setpagesize=false,%リンクの設定
bookmarks=false,bookmarksopen=false,bookmarksnumbered=true,%しおりは文字化けするので使用しない
pdfborder={0 0 0},linkcolor=blue,citecolor=blue,colorlinks=true]{hyperref} 
\graphicspath{{./figures/}}	%画像ファイルが入ったフォルダを指定。

\newcommand{\diff}{\mathrm{d}}	%微分の「d」の定義
\newcommand{\Wass}{\mathcal{W}}

\begin{document}
\begin{CJK*}{GB}{} % Use default fonts from CJK (see below)
\title{Thermodynamically Optimal Information Gain in
Finite Time Measurement}
\author{Rihito Nagase}
\affiliation{Department of Applied Physics, The University of Tokyo, 7-3-1 Hongo, Bunkyo-ku, Tokyo 113-8656, Japan}
\author{Takahiro Sagawa}
\affiliation{Department of Applied Physics, The University of Tokyo, 7-3-1 Hongo, Bunkyo-ku, Tokyo 113-8656, Japan}
\affiliation{Quantum-Phase Electronics Center (QPEC), The University of Tokyo, 7-3-1 Hongo, Bunkyo-ku, Tokyo 113-8656, Japan}
\begin{abstract}
The tradeoff relation between speed and cost is a central issue in designing fast and efficient information processing devices. We derive an achievable bound on thermodynamic cost for obtaining information through finite-time (non-quasi-static) measurements. Our proof is based on optimal transport theory, which enables us to identify the explicit protocol to achieve the obtained bound. Moreover, we demonstrate that the optimal protocol can be approximately implemented by an experimentally feasible setup with quantum dots. Our results would lead to design principles of high-speed and low-energy-cost information processing.
\end{abstract}

\maketitle
\end{CJK*}
%\textit{Introduction.---}
\section{Introduction}\label{sec_1}
Thermodynamics of information has revealed the fundamental bounds on energy cost for information processing such as measurement, feedback, and information erasure \cite{MaxwellDemon,Parrondo_Horowitz_Sagawa_NatPhys_2015,Sagawa_Ueda_PhysRevLett.102.250602_2009,Sagawa_Ueda_PhysRevLett.104.090602_2010,Sagawa_Ueda_PhysRevLett.109.180602_2012,Horowits_Esposito_Missimiliano_PhysRevX.4.031015_2014,Toyabe_Sagawa_Ueda_NatPhys_2010,Berut_Nature_2012,Ciliberto_PhysRevX.7.021051_2017}, which has been applied to a variety of systems including biological systems \cite{Barato_NewJPhys_2014,Sartori_PLOS_2014,Ito_Sagawa_NatCommun_2015} and CMOS devices \cite{Freitas_PhysRevX.11.031064_2021,Freitas_PhysRevLett.129.120602_2022}. The fundamental thermodynamic bounds can be achieved in the quasi-static limit, which requires an infinitely long operation time. Then, a crucial problem lies in determining the fundamental bound on the thermodynamic cost of \textit{finite-time} information processing \cite{Zulkowski_PhysRevE.89.052140_2014,Proesmans_PhysRevE.102.032105_2020,Proesmans_PhysRevLett.125.100602_2020,Nakazato_Ito_PhysRevResearch.3.043093_2021,Tanogami_2023_InfoTUR,Fujimoto_2023_GameTheoretical}, which is significant for designing fast and low-power consumption computers.\par
The thermodynamic cost of finite-time processes has been evaluated by the inequalities called the speed limits \cite{Aurell_JStatPhys_2012,Shiraishi_Funo_Saito_PhysRevLett.121.070601_2018,Ito_PhysRevLett.121.030605_2018,Ito_Dechant_PhysRevX.10.021056_2020,Vo_VanVu_Hasegawa_PhysRevE.102.062132_2020,Plata_PhysRevE.101.032129_2020,Falasco_Esposito_PhysRevLett.125.120604_2020,Yoshimura_Ito_PhysRevLett.127.160601_2021}. Among them, optimal transport theory \cite{OptimalTransport} provides the achievable bounds \cite{Jordan_Kinderlehrer_Otto_SIAM_1998,Benamou_Brenier_2000ACF,Maas_JFuncAnaly_2011gradient,Auruell_PhysRevLett.106.250601_2011,Auruell_JStatPhys_2012,Dechant2019thermodynamic,Chen_Georgiou_IEEE_2017matricial,Chen_Georgiou_IEEE_2019stochastic,Fu_Georgiou_Automatica_2021maximal,Miangolarra_IEEE_2022geometry,VanVu_Hasegawa_PhysRevLett.126.010601,Dechant_JPhys_2022,Yoshimura_PRR_2023,Chennakasavalu_PhysRevLett.130.107101_2023,T.V.Vu_PhysRevX.13.011013}, by determining the optimal transport plan which minimizes the thermodynamic cost (i.e., entropy production) for transporting one distribution to another. In a broader context, optimal transport theory has a wide range of applications including image processing \cite{Haker_Zhu_Tannenbaum_Angenent_IJCV_2004}, machine learning \cite{Kolouri_7974883_2017}, and biology \cite{Schiebinger_Cell_2019}. The key concept in this theory is the Wasserstein distance quantifying the distance between probability distributions, which gives the lower bound on the thermodynamic cost when applied to thermodynamics. \par
%In this Letter, we reveal the fundamental bound on the thermodynamic cost required for finite-time measurement by utilizing optimal transport theory. First, we will prove that the lower bound on the entropy production $\Sigma$ required to obtain mutual information $I$ in finite operation time $\tau$ is given by
As shown later (see also Ref. \cite{Nakazato_Ito_PhysRevResearch.3.043093_2021}), the fundamental bound on the thermodynamic cost $\Sigma$ for obtaining mutual information $I$ in finite operation time $\tau$ is given by the form
\begin{equation}
\Sigma\geq I+\Wass f\left(\frac{\Wass}{D\tau}\right),\label{eq_1}
\end{equation}
where $\Wass$ represents the Wasserstein distance, $D$ is the fixed timescale of the time evolution, and $f(x)$ is an increasing function which is determined by the particular choice of timescale $D$. For example, $f(x)=x$ when we choose time-averaged mobility as $D$, while $f(x)=2\tanh^{-1}(x)$ when we choose time-averaged activity as $D$ (see Appendix A for rigorous definitions of mobility and activity).
%where $\langle m\rangle_\tau$ is  which represents the the response of current to the added thermodynamic force, then $f(x)=x$, and if we set $D=\langle a\rangle_\tau$ where $\langle a\rangle_\tau$ is time-averaged activity which represents the average number of transition, then $f(x)=2\tanh^{-1}(x)$ 
The second term on the right-hand side, which represents the finite-time effect, vanishes in the limit of $\tau\to\infty$. In this limit, inequality (\ref{eq_1}) reduces to the fundamental bounds obtained in Refs. \cite{Sagawa_Ueda_PhysRevLett.102.250602_2009,Sagawa_Ueda_PhysRevLett.109.180602_2012}. Here, it is crucial that the information $I$ and the distance $\Wass$ are not independent, and thus the minimization problem for $\Wass$ under given $I$ is yet to be solved. In other words, the fundamental thermodynamic cost is not determined solely by inequality (\ref{eq_1}).\par
In this paper, we solve this problem; We derive the achievable lower bound on $\Wass$ under given $I$, and identifies the explicit protocol to achieve the bound. This provides the achievable speed limit for finite-time measurement processes. To put it another way, we derive the achievable upper bound of $I$ under given $\Wass$, which is the mathematically most nontrivial result of this paper. This determines the truly fundamental cost required for finite-time measurement. Moreover, we demonstrate that the optimal protocol can be approximately implemented by tuning experimentally accessible parameters of two interacting quantum dots. These results push forward finite-time information thermodynamics, and would serve as a design principle that makes high speed compatible with high energy efficiency in information processing. \par
The organization of this paper is as follows. In Section \ref{sec_2}, we describe our setup along with the review of stochastic thermodynamics and speed limits. In Section \ref{sec_3}, we present the main results of this paper. We describe the main mathematical theorem, and then derive its physical consequences: two fundamental bounds on the energy cost of measurement. In Section \ref{sec_4}, we introduce a system of coupled quantum dots as an implementation for our results, and give the approximately optimal measurement protocol. In Section \ref{sec_5}, we summarize the results of this paper and discuss future prospects.

%
%
%
% Setup
%
%
%
%\textit{Setup.---} 
\section{Setup}\label{sec_2}
\subsection{Stochastic thermodynamics}\label{subsec_2A}
We consider classical discrete-state systems denoted as $X$ and $Y$. The entire system is attached to a heat bath at inverse temperature $\beta$, and its time evolution is described by a Markov jump process. $Y$ plays the role of memory and stores information about system $X$ through a measurement from time $t=0$ to $\tau$.\par
Let $x\in\mathcal{X}\coloneqq\{1,2,\cdots,n_X\}$ and $y\in\mathcal{Y}\coloneqq\{1,2,\cdots,n_Y\}$ represent the possible states that $X$ and $Y$ can take, respectively. The joint probability of finding the entire system in state $(x, y)$ at time $t$ is denoted as $p^{XY}_t(x, y)$, and the marginal probability of state $x$ for $X$ (resp. $y$ for $Y$) is denoted as $p^X_t(x)=\sum_yp^{XY}_t(x, y)$ (resp. $p^Y_t(y)=\sum_x p^{XY}_t(x, y)$). Memory $Y$ is initialized to take the value $y = 1$ with probability $1$ at $t = 0$. We also assume that the transition of $X$ does not occur during the measurement, namely, for any different $x'$ and $x$, the transition from state $(x',y')$ to $(x,y)$ does not occur. This assumption of ignoring the back action of measurement is a standard assumption for classical measurements, because there is no fundamental disturbance on the measured system in the classical case.
 For example, if the time scale of the measurement process is sufficiently shorter than that of $X$, or if there are potential walls that inhibit transitions between the states of $X$ (as is the case for the Szilard engine), the measurement back action is irrelevant. \ In this case, $p^X_t(x)$ is fixed to a constant probability $p^X(x)$.\par
 The time evolution of the entire system can be expressed as \cite{Horowits_Esposito_Missimiliano_PhysRevX.4.031015_2014}
\begin{equation}
\begin{split}
\frac{\diff}{\diff t}p_t^{XY}(x,y)=\sum_{y':(y',y)\in\mathcal{N}_x}&\left[R_t(y,y'|x)p^{XY}_t(x,y')\right.\\
&\left.-R_t(y',y|x)p^{XY}_t(x,y)\right].
\end{split}\label{eq_2}
\end{equation}
Here, $R_t(y,y'|x)$ represents the transition rate from state $(x, y')$ to $(x, y)$, and $\mathcal{N}_x$ is the set of pairs of distinct states $(y', y)$ between which the transition under $x$ is allowed. We assume that the possibility of transitions is undirected, so that if $(y',y) \in \mathcal{N}_x$, then $(y,y') \in \mathcal{N}_x$. Let $Q_t^Y(y,y'|x)$ represent the stochastic heat absorbed by $Y$ during the transition from $(x,y')$ to $(x,y)$ at time $t$, where we assume the local detailed balance condition
\begin{equation}
\ln\frac{R_t(y,y'|x)}{R_t(y',y|x)}=-\beta Q_t^Y(y,y'|x).\label{eq_3}
\end{equation}
%As we consider the case wherein the system is attached to a single bath, the stochastic heat can be expressed as $Q_t^Y(y,y'|x)=E_t^{XY}(x,y)-E_t^{XY}(x,y')$, where $E_t^{XY}(x,y)$ denotes the instantaneous energy level of state $(x,y)$ at time $t$.
\par
We next introduce mutual information which characterizes the amount of information that $Y$ obtains from $X$ through the measurement. In general, the Shannon entropy of a probability distribution $p=\{p(i)\}_i$ is given by $S(p)\coloneqq-\sum_ip(i)\ln p(i)$,which defines the mutual information between $X$ and $Y$ at time $\tau$ as
\begin{equation}
I_\tau^{X:Y}\coloneqq S(p_\tau^X)+S(p_\tau^Y)-S(p_\tau^{XY}).\label{eq_4}
\end{equation}
If $X$ and $Y$ are uncorrelated, $I_\tau^{X:Y}=0$. Otherwise, $I_\tau^{X:Y}>0$. 
%An example of obtaining positive $I_\tau^{X:Y}$ through a measurement process is shown in Fig. \ref{fig_1} (a).
%\begin{figure}[tbp]
%\centering
%\includegraphics[scale=0.47]{fig_meas6.pdf}
%\caption{An example of the measurement process. $X$ stochastically takes values $x=1,\ 2$, and $Y$ takes $y=1$ with probability $1$ at time $t=0$. After the measurement, the probability of $y=x$ is high for all $x$ (the probability of $y\neq x$ corresponds to measurement error). Thus, $Y$ can obtain $I_\tau^{X:Y}(>0)$.}
%\label{fig_1}
%\end{figure}
\par
We introduce the entropy production, which quantifies the dissipation due to irreversibility of the process, as
\begin{equation}
\Sigma_\tau^{XY}\coloneqq
S\left(p_\tau^{XY}\right)-S\left(p_0^{XY}\right)-\beta Q_\tau^Y.\label{eq_5}
\end{equation}
Here, $Q_\tau^Y\coloneqq\int_0^\tau\sum_{x,y,y'}Q_t^Y(y,y'|x)R_t(y,y'|x)p_t^{XY}(x,y')\diff t$ denotes the total heat absorbed by $Y$ until time $\tau$. We note that $Q_\tau^Y$ equals the total heat absorbed by the entire system, because the transitions of $X$ do not occur during the measurement. Using the entropy production, the second law of thermodynamics is given by $\Sigma_\tau^{XY}\geq0$, which is equivalent to the conventional second law of stochastic thermodynamics\cite{Seifert2012stochastic}. We can also introduce the thermodynamic cost for the memory $Y$ as  $\Sigma_\tau^Y\coloneqq S\left(p_\tau^Y\right)-S\left(p_0^Y\right)-\beta Q_\tau^Y$. Equivalently, $\Sigma_\tau^Y$ can be expressed as $\Sigma_\tau^Y=\beta(W^Y-\Delta F^Y)$, where $W^Y$ is the work done on $Y$, and $\Delta F^Y$ is the change in non-equilibrium free energy of $Y$ \cite{StochasticThermo}, implying that the second law gives the fundamental bounds on the energy cost. Using $\Sigma_\tau^{Y}$ and $I_\tau^{X:Y}$, the entropy production can be decomposed as \cite{Sagawa_Ueda_PhysRevLett.102.250602_2009}
\begin{equation}
\Sigma_\tau^{XY}=\Sigma_\tau^{Y}-I_\tau^{X:Y}.\label{eq_6}
\end{equation}
This decomposition and the second law lead to $\Sigma_\tau^Y\geq I_\tau^{X:Y}$, which is equivalent to inequality (\ref{eq_1}) in the limit of $\tau\to\infty$.\par
%
%
%
% Wasserstein distance
%
%
%
%\textit{Wasserstein distance.---}
\subsection{Wasserstein distance and speed limit}\label{subsec_2B}
We introduce the Wasserstein distance between probability distributions $p_0^{XY}$ and $p_\tau^{XY}$ (hereafter in this paragraph, we abbreviate the superscript $XY$), defined as \cite{T.V.Vu_PhysRevX.13.011013}
\begin{equation}
\begin{split}
&\Wass(p_0,p_\tau)\\
&\ \coloneqq\min_{\pi\in\Pi(p_0,p_\tau)}\sum_{x,y,y'}d(y,y'|x)\pi[(x,y),(x,y')].
\end{split}\label{eq_7}
\end{equation}
Here, $\pi[(x,y),(x,y')]\ (\geq0)$ can be interpreted as the amount of probability sent from state $(x,y')$ to $(x,y)$ in a transport plan $\pi$, and $\Pi(p_0,p_\tau)$ is the set of all transport plans from $p_0$ to $p_\tau$ satisfying $\sum_{y\in\mathcal{Y}}\pi[(x,y),(x,y')]=p_0(x,y')$ and $\sum_{y'\in\mathcal{Y}}\pi[(x,y),(x,y')]=p_\tau(x,y)$. The coefficient $d(y,y'|x)$ represents the minimum number of transitions required to go from state $(x,y')$ to $(x,y)$, defining the transport cost of the distribution.
%(an example is shown in Fig. \ref{fig_1} (b))
\par
When $p_0^{XY}$ is transformed into $p_\tau^{XY}$ obeying Eq. (\ref{eq_2}), the entropy production is bounded as \cite{T.V.Vu_PhysRevX.13.011013}
\begin{equation}
\Sigma_\tau^{XY}\geq \Wass\left(p_0^{XY},p_\tau^{XY}\right)f\left(\frac{\Wass\left(p_0^{XY},p_\tau^{XY}\right)}{D\tau}\right),\label{eq_8}
\end{equation}
where $f(x)=x$ if $D$ is time-averaged mobility $\langle m\rangle_\tau$, and $f(x)=2\tanh^{-1}(x)$ if $D$ is time-averaged activity $\langle a\rangle_\tau$. There exist protocols $\{R_t(y,y'|x)\}$ that achieve the equalities for each choice of $D$, but the construction methods are different. When $D=\langle m\rangle_\tau$, the optimal protocol $\{R_t(y,y'|x)\}$ is determined by the condition that probabilities are transported under a uniform and constant thermodynamic force along the optimal transport plan from $p_0^{XY}$ to $p_\tau^{XY}$. For $D=\langle a\rangle_\tau$, the condition for the optimal protocol $\{R_t(y,y'|x)\}$ is that probabilities are transported under uniform and constant activity and probability current along the optimal transport plan.\par
By substituting the decomposition of Eq. (\ref{eq_6}) into inequality (\ref{eq_8}), we obtain inequality (\ref{eq_1}) by identifying $\Sigma=\Sigma_\tau^{Y}$, $I=I_\tau^{X:Y}$, and $\Wass=\Wass\left(p_0^{XY},p_\tau^{XY}\right)$. We note that $p_\tau^{XY}$ is fixed and thus mutual information $I_\tau^{X:Y}$ is fixed in the present setup. Therefore, the optimization of $\Sigma_\tau^Y$ is equivalent to the optimization of $\Sigma_\tau^{XY}$, which is achieved by 
%s $\{R_t(y',y|x)\}$ under fixed $p_0^{XY}$ and $p_\tau^{XY}$, only $\Sigma_\tau^{XY}$ changes, where the total energy cost is expressed as $\Sigma_\tau^Y=I_\tau^{X:Y}+\Sigma_\tau^{XY}$, which is equivalent to (\ref{eq_6}). The minimum value of $\Sigma_\tau^{XY}$ is given by inequality (\ref{eq_8}), achieved by a 
a certain protocol $\{R_t(y',y|x)\}$. We can then achieve the equality of (\ref{eq_1}) by the same protocol. In this sense, the right-hand side of inequality (\ref{eq_1}), which is expressed as the function of mutual information and Wasserstein distance, provides the minimum energy cost required to transform a fixed initial distribution into the fixed final distribution through measurement processes.\ We note that the Wasserstein distance in this work is the Wasserstein-1 distance, while a bound similar to (\ref{eq_1}) has been obtained in Ref. \cite{Nakazato_Ito_PhysRevResearch.3.043093_2021} by using Wasserstein-2 distance.
%\section{The control protocol for coupled quantum dots}
%
%
%
% Main Theorem
%
%
%
%\textit{Main Theorem.---}
\section{Fundamental bound on the energy cost}\label{sec_3}
\subsection{Main theorem}\label{sec_3A}
As mentioned before, the fundamental bounds on the energy cost for obtaining a certain amount of information cannot be determined solely from inequality (\ref{eq_1}). Therefore, we first solve the dual problem that reveals the upper bound on $I_\tau^{X:Y}$ under fixed $\Wass\left(p_0^{XY},p_\tau^{XY}\right)$. This is the most nontrivial part of our study and can be summarized in the following theorem.
\begin{thm}
For any fixed $\Wass=\Wass\left(p_0^{XY},p_\tau^{XY}\right)$, there exists a probability distribution $\tilde{p}_\Wass^X$ (explicitly defined below) such that
\begin{equation}
I_\tau^{X:Y}\leq S\left(p^X\right)-(1-\Wass)S\left(\tilde{p}^X_\Wass\right)\eqqcolon I_{p^X}(\Wass).\label{eq_9}
\end{equation}
\end{thm}
\begin{figure*}[tbp]
\centering
\includegraphics[keepaspectratio, scale=0.53]
      {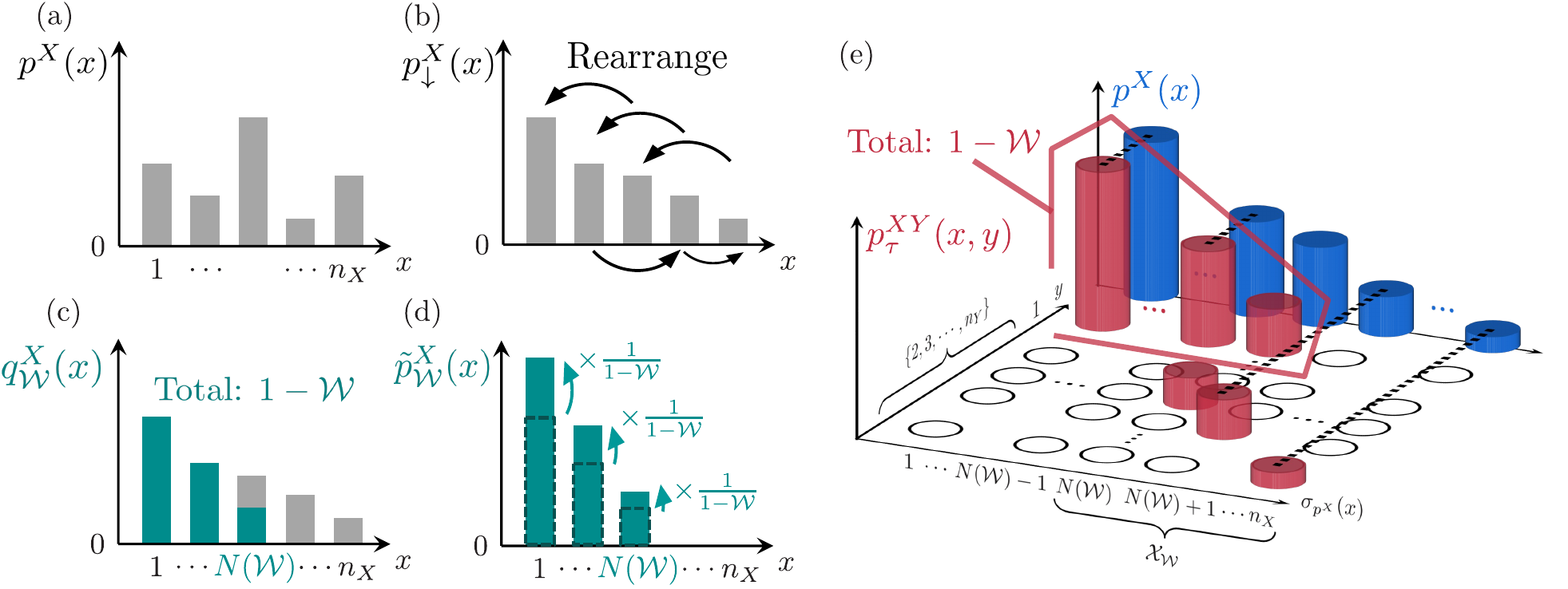}
 \caption{(a) The histogram of $p^X$. (b) The histogram of $p_\downarrow^X$ obtained by rearranging $p^X$ in the descending order. (c) The construction of $q_\Wass^X$. The gray bars represent the histogram of $p_\downarrow^X$, and the green bars represent that of $q_\Wass^X$. (d) The construction of $\tilde{p}_\Wass^X$. The green dots show the histogram of $q_\Wass^X$, and the green bars are that of $\tilde{p}_\Wass^X$, obtained by multiplying each component of $q_\Wass^X$ by $1/(1-\Wass)$. (e) Illustration of the optimal distribution $p_\tau^{XY}$ and $p^X$. The horizontal axes represent the states, and the vertical axis represents the probability. The thick dashed lines indicate equal probabilities.}
 %\caption{(a) The histogram of $p^X$. (b) Definition of $\tilde{p}_\Wass^X$ and $N(\Wass)$ for $\Wass<1$. First, we obtain a polygonal line (green dashed line) by truncating the Lorenz curve of $p^X$ by the value $1-\Wass$. Then, we define $\tilde{p}_\Wass^X$ as a probability distribution corresponding to the Lorenz curve obtained by scaling the green dashed line by a factor of $1/(1-\Wass)$ (green solid line). $N(\Wass)$ is defined as the minimum value of $x$ which satisfies $\sum_{x'=1}^xp_\downarrow^X(x')>1-\Wass$ (if $\Wass=0$, $N(\Wass)\coloneqq n_X$). (c) Definition of $\tilde{p}_\Wass^X$ and $N(\Wass)$ for $\Wass\geq1$. (d) Comparison of Lorenz curves of $\tilde{p}_\Wass^X$ and $\tilde{p}_{\Wass'}^X$ when $\Wass'>\Wass$. The Lorenz curve of $\tilde{p}_{\Wass'}^X$ lies above that of $\tilde{p}_{\Wass}^X$. (e) Illustration of the optimal distribution $p_\tau^{XY}$ and $p^X$. The horizontal axes represent the states, and the vertical axis represents the probability. The thick dashed lines indicate equal probabilities.}
 \label{fig_2}
\end{figure*}
%\begin{figure}[tbp]
%\centering
%\includegraphics[keepaspectratio, scale=0.5]
%      {fig_result3.pdf}
% \caption{(a) The Lorenz curve of $p^X$. (b) Definition of $\tilde{p}_\Wass^X$ and $N(\Wass)$ for $\Wass<1$. First, we obtain a polygonal line (green dashed line) by truncating the Lorenz curve of $p^X$ by the value $1-\Wass$. Then, we define $\tilde{p}_\Wass^X$ as a probability distribution corresponding to the Lorenz curve obtained by scaling the green dashed line by a factor of $1/(1-\Wass)$ (green solid line). $N(\Wass)$ is defined as the minimum value of $x$ which satisfies $\sum_{x'=1}^xp_\downarrow^X(x')>1-\Wass$ (if $\Wass=0$, $N(\Wass)\coloneqq n_X$). (c) Definition of $\tilde{p}_\Wass^X$ and $N(\Wass)$ for $\Wass\geq1$. (d) Comparison of Lorenz curves of $\tilde{p}_\Wass^X$ and $\tilde{p}_{\Wass'}^X$ when $\Wass'>\Wass$. The Lorenz curve of $\tilde{p}_{\Wass'}^X$ lies above that of $\tilde{p}_{\Wass}^X$. 
% (e) Illustration of the optimal distribution $p_\tau^{XY}$ and $p^X$. The horizontal axes represent the states, and the vertical axis represents the probability. The thick dashed lines indicate equal probabilities.
% }
% \label{fig_2}
%\end{figure}
We can graphically define $\tilde{p}_\Wass^X$ starting from the histogram of probability distribution $p^X$ which is shown in Fig. \ref{fig_2} (a) (see Appendix A for the fully rigorous definition of $\tilde{p}_\Wass^X$ and the proof of the theorem). First, let $p^X_\downarrow$ be the probability distribution obtained by rearranging $p^X$ in the descending order $p_\downarrow^X(1)\geq p_\downarrow^X(2)\geq\cdots\geq p_\downarrow^X(n_X)$ (Fig. \ref{fig_2} (b)). Next, construct the distribution $q_\Wass^X$ by extracting probabilities from the distribution $p_\downarrow^X$ in descending order until their total sums up to $1-\Wass$ (Fig. \ref{fig_2} (c)). Here, we define $N(\Wass)$ as the minimal $x$ such that $q_\Wass^X(x)<p_\downarrow^X(x)$. Finally, we obtain the probability distribution $\tilde{p}_\Wass^X$ by multiplying each component of $q_\Wass^X$ by $1/(1-\Wass)$ (Fig. \ref{fig_2} (d)). We define $\tilde{p}_\Wass^X$ as $\tilde{p}_\Wass^X(1)=1$ for $\Wass\geq1$.\par
From this definition, it can be seen that as $\Wass$ increases, the nonzero components of $\tilde{p}_\Wass^X$ concentrate on smaller values of $x$. Since the Shannon entropy takes smaller value for skewed probability distribution, it can be intuitively understood that $S\left(\tilde{p}_\Wass^X\right)$ is a decreasing function of $\Wass$ (for a rigorous proof, see Appendix A). Therefore, the right-hand side $I_{p^X}(\Wass)$ of inequality (\ref{eq_9}) is an increasing function of $\Wass$, indicating that a larger Wasserstein distance $\Wass$ allows $Y$ to obtain more mutual information from $X$.
\subsection{The conditions for achieving our bound}\label{subsec_3B}
We next discuss the optimal protocol that achieves the equality in (\ref{eq_9}). First, The conditions that memory $Y$ should satisfy for achieving the equality is listed as follows: 
\begin{itemize}
\item[(C1)] $n_Y-1\geq n_X-N(\Wass)+1$,
\item[(C2)] For any $x$ and $y\ (\neq1)$, the direct transition from $(x,1)$ to $(x,y)$ is allowed (i.e., $(1,y)\in\mathcal{N}_x$).
\end{itemize}
%We note that (C2) is the sufficient condition, whose relaxation is discussed in Supplemental Materials. \par
We can then achieve the equality by setting the final distribution $p_\tau^{XY}$ to an optimal distribution. We here only focus on describing the construction protocol itself; the proof for its optimality is given in Appendix A.\par
As preparation, we define some notations. First, denote the permutation that rearranges the states \( x \) in descending order according to the probability distribution \( p^X \) as \( \sigma_{p^X} \). \( \sigma_{p^X}(x) = n \) implies that \( x \) has the \( n \)-th largest value of \( p^X(x) \). Next, define \( \mathcal{X}_\Wass \) as the set of states \( x \) that come after the \( N(\Wass) \)-th position when sorted in descending order of $p^X$, i.e., $\mathcal{X}_\Wass\coloneqq\{x| N(\Wass)\leq\sigma_{p^X}(x)\leq n_X\}$. Finally, for each \( x \in \mathcal{X}_\Wass \), select a corresponding \( y_x \) from the set $\{y|y\neq1\}$, ensuring that each \( y_x \) has one-to-one correspondence with each \( x \).\par
Then, the optimal final distribution is given by
\begin{align}
&p_\tau^{XY}(x,y)\notag\\
&\quad=
\begin{cases}
  q_\Wass^X\left(\sigma_{p^X}(x)\right), & y=1, \\
  p^X\left(\sigma_{p^X}(x)\right)-q_\Wass^X\left(\sigma_{p^X}(x)\right), & y=y_x,\ x\in\mathcal{X}_\Wass, \\
  0, & \mathrm{otherwise.}
\end{cases}\label{eq_10}
\end{align}
This distribution is illustrated in Fig. \ref{fig_2} (e). From Eq. (\ref{eq_10}) and Fig. \ref{fig_2} (e), it is understood that the optimal final distribution is constructed by extracting probabilities from the largest values of \( p^X(x) \) in sequence until the total reaches \( \mathcal{W} \), assigning these probabilities to different \( y_x \) for each \( x \), and arranging the rest along \( y=1 \). We here note that condition (C1) ensures the one-to-one correspondence between $x\in\mathcal{X}_\Wass$ and $y\in\{y|y\neq1\}$. The condition (C2) ensures $d(y,1|x)=1$, which means that there is no increase in the Wasserstein distance due to the passage of additional states during the transport from $p_0^{XY}$ to the optimal distribution $p_\tau^{XY}$.% Under this condition, $0\leq\Wass\leq1$ holds, which ensures that each component in Eq. (\ref{eq_10}) satisfies $0\leq p_\tau^{XY}(x,y)\leq1$.
\par
%
%
%
% Speed limit
%
%
%
%\textit{Speed limit.---}
\subsection{Speed limit for fixed mobility}\label{subsec_3C}
From the Main Theorem, we can obtain the fundamental bound on the thermodynamic cost $\Sigma^{Y}_\tau$ for obtaining mutual information $I_\tau^{X:Y}$ in the measurement process. We first note that $I_{p^X}(\Wass)$ is strictly monotonically increasing for $0\leq\Wass\leq1-p_\downarrow^X(1)$ and takes a constant value $S\left(p^X\right)$ for $\Wass\geq1-p_\downarrow^X(1)$ (see Supplemental Materials for the proof). Thus we can define the inverse function $\Wass_{p^X}:[0,S(p^X)]\to[0,1-p_\downarrow^X(1)]$ for the fixed $p^X$ as $I_{p^X}\left(\Wass_{p^X}(I)\right)=I$.
%described in Fig. \ref{fig_3}. 
%\begin{figure}[tbp]
%\centering
%\includegraphics[scale=0.5]{fig_inv.pdf}
%\caption{Construction of the inverse function$\Wass_{p^X}$ of $I_{p^X}$. For $I\in[0,S\left(p^X\right)$], we define $\Wass_{p^X}(I)\in[0,1-p_\downarrow^X(1)]$ as $I=I_{p^X}\left(\Wass_{p^X}(I)\right)$.}
%\label{fig_3}
%\end{figure}
From this and inequality (\ref{eq_1}) (or equivalently, inequality (\ref{eq_8})), we obtain the speed limit of the measurement process for fixed mobility
\begin{equation}
\Sigma_\tau^Y\geq I_\tau^{X:Y}+\frac{{\Wass_{p^X}\left(I_\tau^{X:Y}\right)}^2}{\tau\langle m\rangle_\tau},\label{eq_11}
\end{equation}
which gives a lower bound on the thermodynamic cost for obtaining a given amount of information $I_\tau^{X:Y}$. Inequality (\ref{eq_11}) specifies the optimal value of $\Wass$ depending on $I_\tau^{X:Y}$. 
%Compared to inequality (\ref{eq_1}) In terms of physical significance, inequality (\ref{eq_11}) is the main result of this Letter.
\par
The equality in (\ref{eq_11}) can be achieved with finite operation time $\tau$ by simultaneously achieving the equalities in (\ref{eq_9}) and (\ref{eq_8}). This is possible as follows. First, the equality in (\ref{eq_9}) can be achieved by preparing a memory $Y$ that satisfies conditions (C1) and (C2) and setting $p_\tau^{XY}$ according to Eq. (\ref{eq_10}). We can then construct $\{R_t(y, y'|x)\}$ that achieves the equality in (\ref{eq_8}) for the fixed $p_0^{XY}$ and $p_\tau^{XY}$ by transporting probabilities under uniform and constant thermodynamic force along the optimal transport plan. \par
%
%
%
% Activity固定のSpeed Limit
%
%
%
\subsection{Speed limit for fixed activity}\label{subsec_3D}
When we choose time-averaged activity $\langle a\rangle_\tau$ as $D$, the second term of the right hand side of inequality (\ref{eq_1}) becomes $2\Wass\tanh^{-1}\left[\Wass/(\tau\langle a\rangle_\tau)\right]$, which is an increasing function of $\Wass$ for any fixed $\langle a\rangle_\tau$ and $\tau$. Therefore, we can also apply the Main Theorem to this case and obtain another speed limit
\begin{equation}
\Sigma_\tau^Y\geq I_\tau^{X:Y}+2\Wass_{p^X}\left(I_\tau^{X:Y}\right)\tanh^{-1}\frac{\Wass_{p^X}\left(I_\tau^{X:Y}\right)}{\tau\langle a\rangle_\tau},\label{eq_act}
\end{equation}
which gives the lower bound on the thermodynamic cost for obtaining information $I_\tau^{X:Y}$ for fixed activity $\langle a\rangle_\tau$.\par
The equality in (\ref{eq_act}) can be achieved with finite operation time $\tau$ by simultaneously achieving the equalities in (\ref{eq_9}) and (\ref{eq_8}), i.e., by preparing a memory $Y$ that satisfies conditions (C1) and (C2), setting $p_\tau^{XY}$ according to Eq. (\ref{eq_10}), and transporting probabilities under uniform and constant activity and probability current along the optimal transport plan from $p_0^{XY}$ to $p_\tau^{XY}$.
%
% 
%
% Example
%
%
%
\section{Example: Double quantum dots}\label{sec_4}
\subsection{Setup}\label{subsec_4A}
We next give an experimentally feasible setup that approximately achieves the minimum energy cost determined by inequality (\ref{eq_11}). We consider two coupled single-level quantum dots attached to reservoirs \cite{Sanchez_PhysRevLett.104.076801_2010,Schaller_PhysRevB.82.041303_2010,Bulnes_PhysRevB.84.165114_2011,Kutvonen_PhysRevE.93.032147_2016} (Fig. \ref{fig_4} (a)). 
\begin{figure}[tbp]
\centering
\includegraphics[scale=0.47]{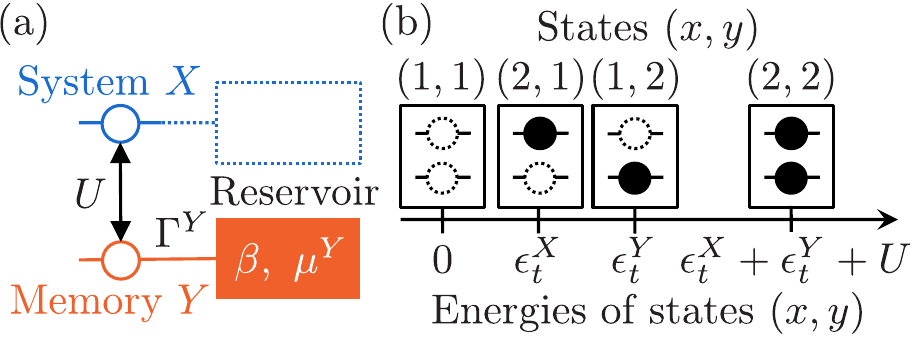}
\caption{(a) A schematic illustration of coupled quantum dots. (b) Energy diagram of the total system.}
\label{fig_4}
\end{figure}
One is the memory $Y$ with energy $\epsilon_t^Y$ at time $t$, and the other is the measured system $X$ with energy $\epsilon_t^X$ at time $t$. $X$ (resp. $Y$) takes states $x=1,\ 2$ (resp. $y=1,\ 2$) corresponding to the unoccupied and occupied states, respectively. Coulomb energy $U$ takes place when both dots are occupied (see Fig. \ref{fig_4} (b)).\par
To prevent transitions of $X$, let $X$ not be in contact with a reservoir during the measurement, while $Y$ attaches to the reservoir at inverse temperature $\beta$ and chemical potential $\mu^Y$ with coupling strength $\Gamma^Y$, whose Fermi distribution is denoted as $f^Y(\epsilon)\coloneqq\left[1+e^{\beta(\epsilon-\mu^Y)}\right]^{-1}$. It is possible to effectively prevent transitions of system $X$ caused by the reservoir without physically detaching $X$ from the reservoir, but by adjusting the coupling strength $\Gamma^X$ that depends on the tunneling rate between the reservoir and the quantum dot. Since the tunneling rate can be tuned by adjusting the gate voltage of the reservoir, it is feasible to make $\Gamma^X$ significantly smaller than $\Gamma^Y$.\ The dynamics of the total system obeys Eq. (\ref{eq_2}), where the transition rates are given by
\begin{equation}
\begin{split}
%R_t((2,y),(1,y))&=\Gamma^Xf^X\left(\epsilon_t^X+\delta_{y,2}U\right),\\
%R_t((1,y),(2,y))&=\Gamma^X\left[1-f^X\left(\epsilon_t^X+\delta_{y,2}U\right)\right],\\
R_t(2,1|x)&=\Gamma^Yf^Y\left(\epsilon_t^Y+\delta_{x,2}U\right),\\
R_t(1,2|x)&=\Gamma^Y\left[1-f^Y\left(\epsilon_t^Y+\delta_{x,2}U\right)\right].
\end{split}\label{eq_12}
\end{equation}
Equality (\ref{eq_12}) satisfies the local detailed balance condition (\ref{eq_3}).\par
We next consider the conditions (C1) and (C2). First, in inequality (\ref{eq_11}), $\Wass\in[0,1-p_\downarrow^X(1)]$ holds, implying $N(\Wass)\geq2$. Therefore, (C1) is satisfied if $n_Y\geq n_X$, which is true in this example because $n_X=n_Y=2$. The condition (C2) is also satisfied because the transitions between $y = 1$ and $2$ are allowed regardless of $x$.
\subsection{Approximately optimal measurement protocol for fixed mobility}\label{subsec_4B}
First, we consider the case where the mobility $\langle m\rangle_\tau$ is fixed. We fix $p^X(1)=p\ (\leq1/2)$ and $p^X(2)=1-p$, and measure the states of $X$ in finite time $\tau$ by the memory $Y$ which initially takes $y=1$. Denoting the probability distribution $p_t^{XY}$ by a matrix $[p_t^{XY}(x,y)]_{x,y}$, for $\Wass\in[0,1-p_\downarrow^X(1)]=[0,p]$, the initial and final optimal distributions are given by
\begin{equation}
p_0^{XY}=\left[\begin{array}{cc}
p & 1-p \\
0 & 0
\end{array}\right],\quad
p_\tau^{XY}=\left[\begin{array}{cc}
p-\Wass & 1-p \\
\Wass & 0
\end{array}\right].\label{eq_13}
\end{equation}
These probability distributions achieve the equality in (\ref{eq_9}). Then, the minimum energy cost determined by inequality (\ref{eq_11}) is achieved by the optimal protocol $\{R_t(y,y'|x)\}$ that achieves the equality of (\ref{eq_8}) for $p_0^{XY}$ and $p_\tau^{XY}$.\par
In this example, we can approximately implement the optimal protocol as follows. First, let $\beta U$ be sufficiently large to prevent the transition from state $(2,1)$ to $(2,2)$ (corresponding to the second columns of Eq. (\ref{eq_13})). Then, wait for a positive probability $p\Delta$ to be stored in the state $(1,2)$. This process requires non-optimal energy cost, which can be made arbitrarily small by setting $\Delta$ to arbitrarily small value. Finally, send the required probability from state $(1,1)$ to $(1,2)$ with constant thermodynamic force $F$ by manipulating the energy $\epsilon_t^Y$ (corresponding to the first columns of Eq. (\ref{eq_13})). See Supplemental Materials for the detailed protocol. We here note that we cannot construct such a protocol when $\Wass=p$, but we can set $\Wass$ arbitrarily close to $p$. The approximately optimal protocol converges to the exactly optimal one in the limits $\beta U\to\infty$ and $\Delta\to0$.\par
\begin{figure*}[tbp]
\centering
\includegraphics[scale=0.7]{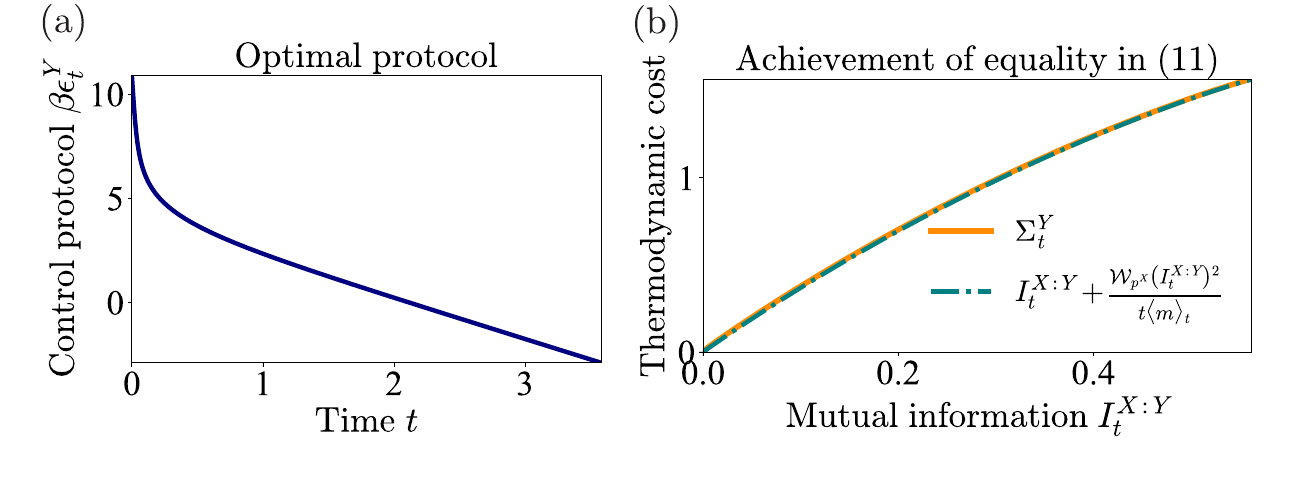}
\caption{Numerical results on the example with two quantum dots. The parameters are set to $\beta U=20$, $F=4$, $\mu^Y=0$, $\Gamma^Y=2$, $\Delta=0.001$. (a) The optimal protocol $\epsilon_t^Y$ scaled by $\beta$, which is obtained by numerically solving the condition for maintaining constant thermodynamic force $F$. (b) Comparison of the left- and right-hand side of inequality (\ref{eq_11}). We see that the protocol shown in (a) approximately (but almost exactly) achieves the equality in (\ref{eq_11}) for any time $t$.}
\label{fig_5}
\end{figure*}
The numerical demonstration of the control protocol of $\epsilon_t^Y$ for $p = 0.25$ and $\Wass=0.999p$ is shown in Fig. \ref{fig_5} (a), and Fig. \ref{fig_5} (b) compares the energy cost $\Sigma_t^Y$ at time $t$ with the right-hand side of (\ref{eq_11}) for $\tau=t$. Figure \ref{fig_5} (b) shows that the equality in (\ref{eq_11}) is approximately (but almost exactly) achieved at any time $t$, and and thus this example can be regarded as the optimal information gain in finite time.
\subsection{Approximately optimal measurement protocol for fixed activity}\label{subsec_4C}
We next consider the case where the activity $\langle a\rangle_\tau$ is fixed. In this case, by varying not only $\epsilon_t^Y$ but also $\Gamma^Y=\Gamma_t^Y$ over time $t$, we can approximately construct an optimal measurement protocol as follows. As a preparation, set the initial and final distributions as Eq. (\ref{eq_13}) and let $\beta U$ be sufficiently large. Then, wait for a positive probability $p\Delta$ to be stored
in the state $(1, 2)$, which accompanies non-optimal energy cost which can be made arbitrarily small by setting $\Delta$ to arbitrarily small value. These parts are the same as the case for fixed mobility.\par
Starting from the distribution $p_t^{XY}(1,2)=p\Delta(>0)$, we can transprot probability from state $(1,2)$ to $(2,2)$ under constant activity $a$ and probability current $J$ by properly controlling $\epsilon_t^Y$ and $\Gamma_t^Y$ (explicit protocol is provided in Appendix C). In this protocol, it is not possible to set $p_\tau^{XY}(1,2)=0$. Therefore, we cannot completely optimize the final distribution for $\Wass=p$. Nevertheless, by adjusting parameters, it is possible to set the final distribution to one arbitrarily close to the optimal distribution in the case $\Wass=p$. This approximately optimal protocol also converges to the exactly optimal one in the limits $\beta U\to\infty$ and $\Delta\to0$.\par
\begin{figure*}[tbp]
\centering
\includegraphics[scale=0.7]{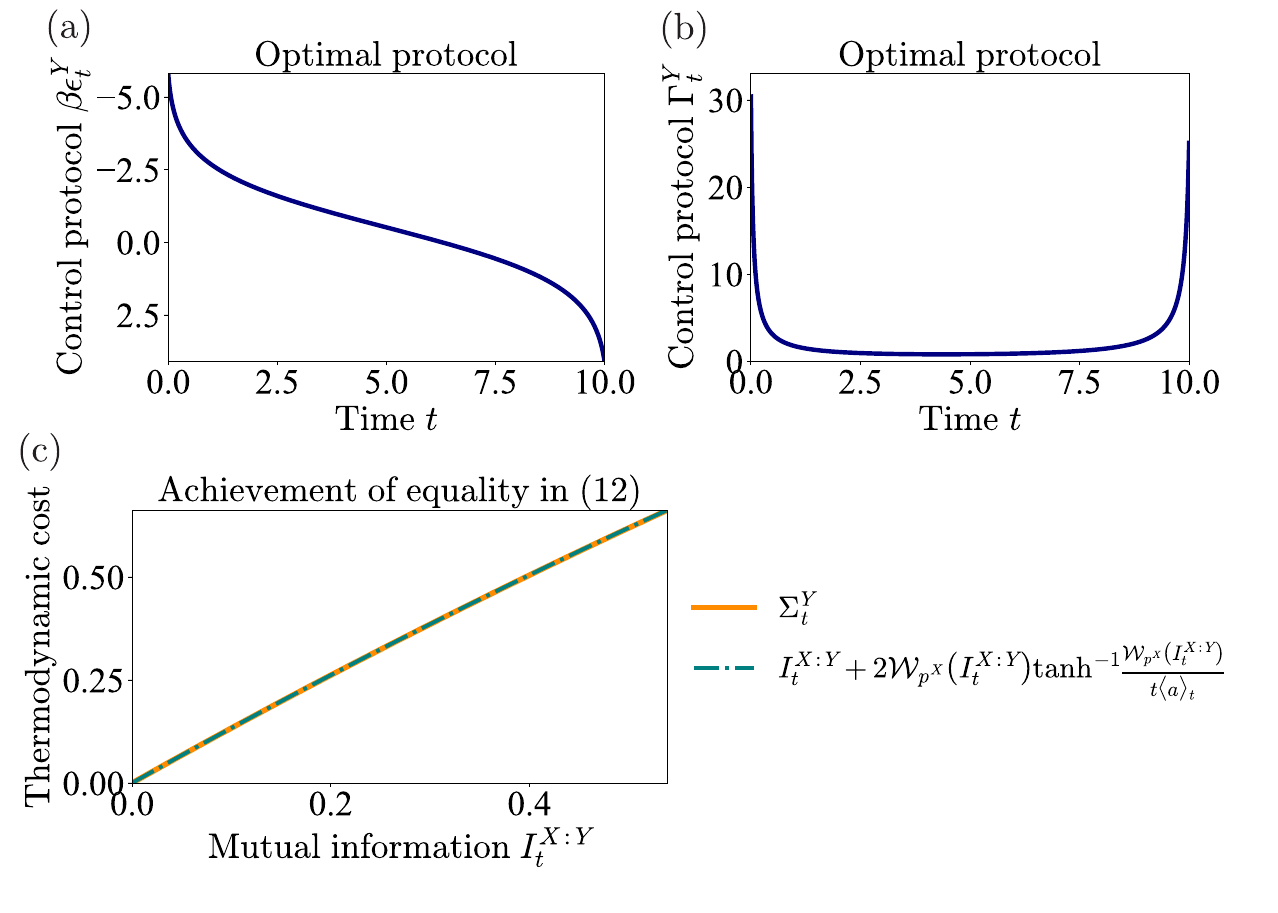}
\caption{Numerical results on the example with two quantum dots. The parameters are set to $\beta U=20$, $a=0.1$, $\tau=10$, $\mu^Y=$, $\Delta=0.005$. (a) The optimal protocol $\epsilon_t^Y$ scaled by $\beta$, which is obtained by numerically solving the condition for maintaining constant activity $a$ and probability current $J$. (b) The optimal protocol $\Gamma_t^Y$. (c) Comparison of the left- and right-hand side of inequality (\ref{eq_act}). We see that the protocols shown in (a) and (b) approximately (but almost exactly) achieves the equality in (\ref{eq_act}) for any time $t$.}
\label{fig_6}
\end{figure*}
The numerical demonstration of the control protocols of $\epsilon_t^Y$ and $\Gamma_t^Y$ for $p = 0.25$ and $\Wass=0.999p$ are shown in Fig. \ref{fig_6} (a) and (b) respectively, and Fig. \ref{fig_6} (c) compares the energy cost $\Sigma_t^Y$ at time $t$ with the right-hand side of (\ref{eq_act}) for $\tau=t$. It is shown that the equality in (\ref{eq_act}) is approximately (but almost exactly) achieved at any time $t$.
%
%
%
% Conclusion
%
%
%
\section{Discussion}\label{sec_5}
In this study, we revealed the fundamental bound on the thermodynamic cost for obtaining a given amount of information in finite time, and a specific protocol for achieving it. The Main Theorem is inequality (\ref{eq_9}) derived from optimal transport theory, which leads to the speed limits (\ref{eq_11}) and (\ref{eq_act}). The equalities can be achieved by properly designing memory and protocols. Moreover, we showed that such an optimal protocol can be approximately implemented by the coupled quantum dots system.\par
The optimal measurement protocol proposed in this study is general and model-independent, and therefore would lead to fast and cost-effective information gain in a wide range of information processing, possibly including calculations using CMOS devices \cite{Freitas_PhysRevX.11.031064_2021,Freitas_PhysRevLett.129.120602_2022,Freitas_PhysRevE.107.014136_2023}. Our result would also be regarded as a step toward clarifying the finite-time effects on thermodynamics in various information processing in terms of optimal transport theory.\par
As future perspectives, it might be possible to generalize our result to discrete infinite cases, but there should be mathematical subtle points to be considered. It is also noteworthy that the generalization the Langevin case is nontrivial and important future issue.
%
%
%
% 謝辞
%
%
%
\section*{ACKNOWLEDGEMENTS}We thank Toshihiro Yada and Yosuke Mitsuhashi for valuable discussions. We are also grateful to Sosuke Ito for insightful comments on the manuscript. R.N. is supported by the World-leading Innovative Graduate Study Program for Materials Research, Industry, and Technology (MERIT-WINGS) of the University of Tokyo. T.S. is supported by Japan Society for the Promotion of Science (JSPS) KAKENHI Grant No. JP19H05796, JST, CREST Grant No. JPMJCR20C1 and JST ERATO-FS Grant No. JPMJER2204. T.S. is also supported by the Institute of AI and Beyond of the University of Tokyo and JST ERATO Grant No. JPMJER2302, Japan. 
%\small
%
%
%
% Refs
%
%
%
\let\oldaddcontentsline\addcontentsline% Store \addcontentsline
\renewcommand{\addcontentsline}[3]{}% Make \addcontentsline a no-op
%apsrev4-2.bst 2019-01-14 (MD) hand-edited version of apsrev4-1.bst
%Control: key (0)
%Control: author (72) initials jnrlst
%Control: editor formatted (1) identically to author
%Control: production of article title (-1) disabled
%Control: page (0) single
%Control: year (1) truncated
%Control: production of eprint (0) enabled

%%%%%%%%%% Merge with supplemental materials %%%%%%%%%%
%\clearpage
%\widetext
\begin{center}
%\textbf{\large APPENDIX}
\end{center}
%%%%%%%%%% Merge with supplemental materials %%%%%%%%%%
%%%%%%%%%% Prefix a "S" to all equations, figures, tables and reset the counter %%%%%%%%%%
\setcounter{equation}{0}
\setcounter{figure}{0}
\setcounter{section}{0}
\setcounter{subsection}{0}
\setcounter{table}{0}
\setcounter{page}{1}
\makeatletter
\renewcommand{\theequation}{A\arabic{equation}}
\renewcommand{\thefigure}{A\arabic{figure}}
\renewcommand{\thesubsection}{\arabic{subsection}}
\renewcommand{\theprop}{A\arabic{prop}}
\renewcommand{\thelem}{A\arabic{lem}}
%\renewcommand{\bibnumfmt}[1]{[A#1]}
%\renewcommand{\citenumfont}[1]{A#1}
%%%%%%%%%% Prefix a "S" to all equations, figures, tables and reset the counter %%%%%%%%%%
\addtocontents{toc}{\protect\setcounter{tocdepth}{0}}
{
%\tableofcontents
}
\begin{appendix}
  \section{Derivation of the main theorem and the speed limit in the main text}\label{sec_A}
%
%
%
% Majorization
%
%
%
\subsection{Majorization}\label{subsec_A1}
First, we introduce the \textit{majorization}, which will be used to obtain an upper bound on the mutual information. Let $p^X$ and $q^X$ be probability distributions on $\mathcal{X}$. We state that $p^X$ majorizes $q^X$, denoted as $q^X\prec p^X$, if
\begin{equation}
\forall x\in\mathcal{X},\ \sum_{x'=1}^{x}q_\downarrow^X(x')\leq \sum_{x'=1}^{x}p_\downarrow^X(x')
\end{equation}
holds [50], where $p_\downarrow^X$ and $q_\downarrow^X$ are the probability distributions obtained by rearranging $p^X$ and $q^X$ in descending order $p_\downarrow^X(1)\geq p_\downarrow^X(2)\geq\cdots\geq p_\downarrow^X(n_X)$ and $q_\downarrow^X(1)\geq q_\downarrow^X(2)\geq\cdots\geq q_\downarrow^X(n_X)$, respectively. The statement that $p^X$ majorizes $q^X$ implies that $q^X$ is more randomly (or uniformly) distributed than $p^X$.\par
If $q^X\prec p^X$, for any convex function $f$, 
\begin{equation}
\sum_{x\in\mathcal{X}}f\left(q^X(x)\right)\leq\sum_{x\in\mathcal{X}}f\left(p^X(x)\right)
\end{equation}
holds. By choosing $f(x)=x\ln x$, we obtain the monotonicity of the Shannon entropy
\begin{equation}
q^X\prec p^X\implies S(q^X)\geq S(p^X).\label{eq_S3}
\end{equation}
%
%
%
% The lower bound on Wasserstein distance
%
%
%
\subsection{The lower bound on the Wasserstein distance}\label{subsec_A2}
The Wasserstein distance between probability distributions $p_0^{XY}$ and $p_\tau^{XY}$ is defined as
\begin{align}
&\Wass\left(p_0^{XY},p_t^{XY}\right)\notag\\
&=\min_{\pi\in\Pi(p_0^{XY},p_t^{XY})}\sum_{x,y,y'}d[(x,y),(x,y')]\pi[(x,y),(x,y')],\label{eq_S6}
\end{align}
where
\begin{align}
&\Pi(p_0^{XY},p_t^{XY})\notag\\
&=\left\{\pi\in\mathbb{R}^{n_Xn_Y\times n_Xn_Y}_{\geq 0}\left|\sum_{y\in\mathcal{Y}}\pi[(x,y),(x,y')]=p_0^{XY}(x,y'),\ \right.\right.\notag\\
&\qquad\left.\left.\sum_{y'\in\mathcal{Y}}\pi[(x,y),(x,y')]=p_t^{XY}(x,y)\right.\right\}.\label{eq_S7}
\end{align}
Taking the initial distribution $p_0^{XY}=p^X(x)\delta_{y,1}$, $\sum_{y}\pi[(x,y),(x,y')]=p_0^{XY}(x,y')=p^X(x)\delta_{y',1}=0$ holds for $y\neq 1$, which implies that
\begin{equation}
\forall y\neq1,\ \pi[(x,y),(x,y')]=0.\label{eq_S6.5}
\end{equation}
Therefore, we obtain
\begin{equation}
p_t^{XY}(x,y)=\sum_{y'}\pi[(x,y),(x,y')]=\pi[(x,y),(x,1)].\label{eq_S7.5}
\end{equation}
By substituting Eqs. (\ref{eq_S6.5}) and (\ref{eq_S7.5}) into Eq. (\ref{eq_S6}), we obtain 
\begin{align}
&\Wass\left(p_0^{XY},p_t^{XY}\right)\notag\\
&=\min_{\pi\in\Pi(p_0^{XY},p_t^{XY})}\sum_{x,y}d[(x,y),(x,1)]\pi[(x,y),(x,1)]\notag\\
&=\sum_{x,y}d[(x,y),(x,1)]p_t^{XY}(x,y),
\end{align}
which means that the Wasserstein distance is obtained only from Eqs. (\ref{eq_S6.5}) and (\ref{eq_S7.5}) without directly solving the minimization problem. We note that
\begin{equation}
d[(x,y),(x,1)]
\begin{cases}
=0, & y=1,\\
\geq 1, & y\neq 1,
\end{cases}\label{eq_S11}
\end{equation}
which yields the following lower bound on the Wasserstein distance:
\begin{align}
  \Wass\left(p_0^{XY},p_t^{XY}\right)&\geq\sum_{x,y:y\neq1}p_t^{XY}(x,y)\notag\\
  &=\sum_{y(\neq1)}p_t^{Y}(y)\notag\\
  &=1-p_t^Y(1)\notag\\
  &\eqqcolon \mathcal{V}_t,
\end{align}
where the equality is achieved if $d[(x,y),(x,1)]=1$ holds for all $y\neq1$. This condition is equivalent to the condition (C1). 
%
%
%
% The upper bound on the mutual information by $\mathcal{V}_\tau$
%
%
%
\subsection{The upper bound on the mutual information by $\mathcal{V}_\tau$}\label{subsec_A3}
In this Subsection, we only consider the distributions at $t=\tau$.  Let $\sigma_{p^X}:\mathcal{X}\to\mathcal{X}$ be the permutation defined as
\begin{equation}
\forall x\in\mathcal{X},p^X(x)=p^X_\downarrow(\sigma_{p^X}(x)),
\end{equation}
which rearranges indexes $x$ in the descending order of $p^X$. We note that the inverse permutation of $\sigma_{p^X}$ satisfies $p^X({\sigma_{p^X}}^{-1}(x))=p_\downarrow^X(x)$. Then, we obtain the following Proposition.
\begin{prop}\label{prop_S1}
For any fixed $\mathcal{V}=\mathcal{V}_\tau$, there exists a probability distribution $\tilde{p}_{\mathcal{V}}^X$ such that
\begin{equation}
I_\tau^{X:Y}\leq S(p^X)-(1-{\mathcal{V}})S(\tilde{p}^X_\mathcal{V})\eqqcolon I_{p^X}({\mathcal{V}}).\label{eq_S14}
\end{equation}
\end{prop}
Here, $\tilde{p}_{\mathcal{V}}^X$ can be defined as follows. First, let $N(\mathcal{V})$ be the minimum integer such that $1-\mathcal{V}<\sum_{x=1}^Np^X_\downarrow(x)$ for $\mathcal{V}>0$, and $N(0)\coloneqq n_X$ for $\mathcal{V}=0$. Then, we define $\tilde{p}^X_\mathcal{V}(x)\coloneqq\delta_{\sigma_{p^X}(x),1}$ for $N(\mathcal{V})=1$ and
\begin{equation}
\tilde{p}^X_\mathcal{V}(x)\coloneqq
\begin{cases}
p^X(x)/(1-\mathcal{V}), & \sigma_{p^X}(x)<N(\mathcal{V}),\\
1-\sum_{\sigma_{p^X}(x)=1}^{N(\mathcal{V})-1}\tilde{p}^X_\mathcal{V}(x), & \sigma_{p^X}(x)=N(\mathcal{V}),\\
0, & \mathrm{otherwise}\label{eq_S15}
\end{cases}
\end{equation}
for $N(\mathcal{V})>1$. This definition is equivalent to the graphical definition of $\tilde{p}_\Wass^X$ given in the main text.
\begin{proof}
We introduce the conditional probability distribution of $X$ given $y\in\mathcal{Y}$ at time $t=\tau$ denoted as $p_\tau^{X|y}$ and the conditional entropy of $X$ given $Y$ at time $t=\tau$ denotes as $S\left(p_\tau^{X|Y}\right)$:
\begin{equation}
p_\tau^{X|y}(x)\coloneqq\frac{p_\tau^{XY}(x,y)}{p_\tau^Y(y)},\quad S\left(p_\tau^{X|Y}\right)\coloneqq\sum_{y\in\mathcal{Y}}p_\tau^Y(y)S\left(p_\tau^{X|y}\right).
\end{equation}
Then, $S\left(p_\tau^{X|Y}\right)=S\left(p_\tau^{XY}\right)-S\left(p_\tau^{Y}\right)$ holds, which yields 
\begin{align}
I_\tau^{X:Y}&=S\left(p^{X}\right)-S\left(p_\tau^{X|Y}\right)\\
&= S\left(p^{X}\right)-\sum_{y\in\mathcal{Y}}p_\tau^Y(y)S\left(p_\tau^{X|y}\right)\\
&\leq S\left(p^{X}\right)-p_\tau^Y(1)S\left(p_\tau^{X|1}\right),\label{eq_S18}
\end{align}
where the equality holds if $S\left(p_\tau^{X|y}\right)=0$ for all $y\neq 1$. If $p_\tau^Y(1)=0$, the right-hand side of inequality (\ref{eq_S18}) is $S(p^X)$. Next, we consider the case $p_\tau^Y(1)\neq0$. Since
\begin{align}
p^X(x)&=\sum_yp_\tau^Y(y)p_\tau^{X|y}(x)\notag\\
&=(1-\mathcal{V})p_\tau^{X|1}(x)+\sum_{y(\neq1)}p_\tau^Y(y)p_\tau^{X|y}(x)
\end{align}
holds, the conditional probability $p_\tau^{X|1}(x)$ satisfies
\begin{equation}
0\leq p_\tau^{X|1}(x)\leq\frac{p^X(x)}{1-\mathcal{V}}.\label{eq_S17}
\end{equation}
We now prove $p_\tau^{X|1}\prec \tilde{p}^X_\mathcal{V}$ by contradiction. Suppose that
\begin{equation}
\exists x\in\mathcal{X},\ \sum_{x'=1}^xp_\tau^{X|1}\left(\sigma_{p_\tau^{X|1}}^{-1}(x')\right)>\sum_{x'=1}^x\tilde{p}^X_\mathcal{V}\left(\sigma_{\tilde{p}^X_\mathcal{V}}^{-1}(x')\right),\label{eq_S20}
\end{equation}
and let $x^*$ be such $x$. If $x^*<N(\mathcal{V})$, it contradicts the inequality obtained by summing inequality (\ref{eq_S17}) over $x$ in the range $\sigma_{p^{X|1}}(x)\leq x^*$:
\begin{align}
\sum_{x=1}^{x^*}p_\tau^{X|1}\left(\sigma_{p_\tau^{X|1}}^{-1}(x)\right)
&\leq\sum_{x=1}^{x^*}\frac{p^X\left(\sigma_{p_\tau^{X|1}}^{-1}(x)\right)}{1-\mathcal{V}}\notag\\
&\leq\sum_{x=1}^{x^*}\frac{p^X\left(\sigma_{p^{X}}^{-1}(x)\right)}{1-\mathcal{V}}\notag\\
&=\sum_{x=1}^{x^*}\frac{p^X\left(\sigma_{\tilde{p}^{X}_\mathcal{V}}^{-1}(x)\right)}{1-\mathcal{V}}\notag\\
&=\sum_{x=1}^{x^*}\tilde{p}^{X}_\mathcal{V}\left(\sigma_{\tilde{p}^{X}_\mathcal{V}}^{-1}(x)\right),
\end{align}
where we used $\sigma_{\tilde{p}^{X}_\mathcal{V}}=\sigma_{p^{X}}$, which follows from the definition of $\tilde{p}^{X}_\mathcal{V}$. If $x^*\geq N(\mathcal{V})$, we obtain
\begin{align}
 \sum_{x=1}^{x^*}p_\tau^{X|1}\left(\sigma_{p_\tau^{X|1}}^{-1}(x)\right)&>\sum_{x=1}^{x^*}\tilde{p}^X_\mathcal{V}\left(\sigma_{\tilde{p}^X_\mathcal{V}}^{-1}(x)\right)\notag\\
 &=\sum_{x=1}^{N(\mathcal{V})}\tilde{p}^X_\mathcal{V}\left(\sigma_{\tilde{p}^X_\mathcal{V}}^{-1}(x)\right)\notag\\
 &=1, 
\end{align}
which contradicts $\sum_{x\in\mathcal{X}}p_\tau^{X|1}(x)=1$. Therefore, the negation of (\ref{eq_S20}) 
\begin{equation}
\forall x\in\mathcal{X},\ \sum_{x'=1}^xp_\tau^{X|1}\left(\sigma_{p_\tau^{X|1}}^{-1}(x')\right)\leq\sum_{x'=1}^x\tilde{p}^X_\mathcal{V}\left(\sigma_{\tilde{p}^X_\mathcal{V}}^{-1}(x')\right)
\end{equation}
is true, which implies $p_\tau^{X|1}\prec \tilde{p}^X_\mathcal{V}$. Finally, from the monotonicity of the Shannon entropy (\ref{eq_S3}), we obtain $S\left(p_\tau^{X|1}\right)\geq S\left(\tilde{p}^X_\mathcal{V}\right)$. Then, substituting this and $p_\tau^Y(1)=1-\mathcal{V}$ into inequality (\ref{eq_S18}) leads to inequality (\ref{eq_S14}).
\end{proof}
%
%
%
% The monotonicity of $I_{p^X}(\mathcal{V})$ and the proof of the main theorem
%
%
%
\subsection{The monotonicity of $I_{p^X}(\mathcal{V})$ and the proof of the Main Theorem}\label{subsec_A4}
Although $\mathcal{V}$ was restricted to $\mathcal{V}\in[0,1]$ in the previous Subsection, we can apply the definitions of $N(\mathcal{V})$ and $\tilde{p}_\mathcal{V}^X$ in the previous Subsection to the case $\mathcal{V}>1$, and define $I_{p^X}(\mathcal{V})$ for any non-negative real number $\mathcal{V}$. Then, for $\mathcal{V}>1$, $N(\mathcal{V})=1$ holds, which implies $\tilde{p}^X_\mathcal{V}(x)=\delta_{\sigma_{p^X}(x),1}$. Thus, $I_{p^X}(\mathcal{V})$ takes constant value $I_{p^X}(\mathcal{V})=S(p^X)-(1-\mathcal{V})\cdot0=S(p^X)$ for $\mathcal{V}>1$. We then obtain the following lemma.
\begin{lem}\label{lem_S1}
$I_{p^X}(\mathcal{V})$ is monotonically increasing for $\mathcal{V}\geq0$, and strictly monotonically increasing for $0\leq \mathcal{V}\leq1-p^X_\downarrow(1)$.
\end{lem}
\begin{proof}
Let $\mathcal{V}$ and $\mathcal{V}'$ satisfy
\begin{equation}
0\leqq \mathcal{V}<\mathcal{V}'<1-p_\downarrow^X(1),
\end{equation}
which implies $N(\mathcal{V}')\geq2$. From this and the definition of $\tilde{p}_\mathcal{V}^X$, for $x\in\mathcal{X}$ such that $\sigma_{p^X}(x)=1$,
\begin{equation}
\tilde{p}^X_{\mathcal{V}'}(x)=\frac{p^X(x)}{1-\mathcal{V}'}\in(0,1)
\end{equation}
holds, which leads to $S\left(\tilde{p}_{\mathcal{V}'}^X\right)>0$. Next, we show $\tilde{p}^X_{\mathcal{V}}\prec\tilde{p}^X_{\mathcal{V}'}$. From the definition of $N(\mathcal{V})$, $N(\mathcal{V})\geq N(\mathcal{V}')$ holds for $\mathcal{V}<\mathcal{V}'$. Therefore, for $x<N(\mathcal{V'})$, we obtain
\begin{align}
&\sum_{x'=1}^{x}\tilde{p}_{\mathcal{V}'}^X\left(\sigma_{\tilde{p}_{\mathcal{V}'}^X}^{-1}(x')\right)-\sum_{x'=1}^{x}\tilde{p}_\mathcal{V}^X\left(\sigma_{\tilde{p}_{\mathcal{V}}^X}^{-1}(x')\right)\notag\\
&=\sum_{x'=1}^{x}\frac{p^X\left(\sigma_{p^X}(x')\right)}{1-\mathcal{V}'}-\sum_{x'=1}^{x}\frac{p^X\left(\sigma_{p^X}(x')\right)}{1-\mathcal{V}}\geq0.\label{eq_S26}
\end{align}
For $x\geq N(\mathcal{V}')$, we get $\sum_{x'=1}^{x}\tilde{p}_{\mathcal{V}'}^X\left(\sigma_{\tilde{p}_{\mathcal{V}'}^X}^{-1}(x')\right)=1$, which yields
\begin{align}
&\sum_{x'=1}^{x}\tilde{p}_{\mathcal{V}'}^X\left(\sigma_{\tilde{p}_{\mathcal{V}'}^X}^{-1}(x')\right)-\sum_{x'=1}^{x}\tilde{p}_\mathcal{V}^X\left(\sigma_{\tilde{p}_{\mathcal{V}}^X}^{-1}(x')\right)\notag\\
&=1-\sum_{x'=1}^{x}\tilde{p}_\mathcal{V}^X\left(\sigma_{\tilde{p}_{\mathcal{V}}^X}^{-1}(x')\right)\geq0.\label{eq_S27}
\end{align}
From (\ref{eq_S26}) and (\ref{eq_S27}), we obtain $\tilde{p}^X_{\mathcal{V}}\prec\tilde{p}^X_{\mathcal{V}'}$. This implies $S(\tilde{p}^X_{\mathcal{V}})\geq S(\tilde{p}^X_{\mathcal{V}'})$. From this and $S\left(\tilde{p}_{\mathcal{V}'}^X\right)>0$, we get
\begin{align}
  I_{p^X}(\mathcal{V}')-I_{p^X}(\mathcal{V})&=(1-\mathcal{V})S(\tilde{p}^X_{\mathcal{V}})-(1-\mathcal{V}')S(\tilde{p}^X_{\mathcal{V}'})\notag\\
  &\geq(\mathcal{V}'-\mathcal{V})S(\tilde{p}^X_{\mathcal{V}'})\notag\\
  &>0.
\end{align}
Thus $I_{p^X}(\mathcal{V})$ is strictly monotonically increasing for $0\leq \mathcal{V}<1-p_\downarrow^X(1)$. Furthermore,
\begin{align}
&\lim_{\mathcal{V}\nearrow1-p_\downarrow^X(1)}I_{p^X}(\mathcal{V})\notag\\
&=S(p^X)-\lim_{\mathcal{V}\nearrow1-p_\downarrow^X(1)}(1-\mathcal{V})S(\tilde{p}_\mathcal{V}^X)\\
&=S(p^X)+\lim_{\mathcal{V}\nearrow1-p_\downarrow^X(1)}\left[p_\downarrow^X(1)\ln\frac{p_\downarrow^X(1)}{1-\mathcal{V}}\right.\notag\\
&\hspace{60pt}\left.+\left(1-\mathcal{V}-p_\downarrow^X(1)\right)\ln\frac{1-\mathcal{V}-p_\downarrow^X(1)}{1-\mathcal{V}}\right]\\
&=S(p^X)\\
&=I_{p^X}\left(1-p_\downarrow^X(1)\right)
\end{align}
holds, which implies that $I_{p^X}(\mathcal{V})$ is strictly monotonically increasing for $0\leq \mathcal{V}\leq1-p_\downarrow^X(1)$. For $\mathcal{V}\geq1-p^X_\downarrow(1)$, $N(\mathcal{V})=1$ holds, and thus $I_{p^X}(\mathcal{V})$ takes the constant value $S(p^X)$. Therefore, $I_{p^X}(\mathcal{V})$ is monotonically increasing for $\mathcal{V}\geq0$.
\end{proof}
We now give the proof of the Main Theorem in the main text. By setting $t=\tau$ in $\Wass(p_0^{XY},p_t^{XY})\geq1-p_t^Y(1)=\mathcal{V}_t$ shown in the Subsection \ref{subsec_A3}, and using the Proposition \ref{prop_S1} and the Lemma \ref{lem_S1}, we obtain
\begin{equation}
I_\tau^{XY}\leq I_{p^X}(\Wass)=S(p^X)-(1-\Wass)S(\tilde{p}_\Wass^X),
\end{equation}
which is equivalent to the Main Theorem (Eq. (9)) in the main text. Now, we define $p_\tau^{XY}$ by Eq. (10) in the main text. From  Eq. (10), we obtain $S\left(p_\tau^{X|y}\right)=0$ for $y\neq1$, $p_\tau^Y(1)=1-\Wass$, and $S\left(p_\tau^{X|1}\right)=S\left(\tilde{p}_\Wass^X\right)$. Therefore, 
\begin{align}
 I_\tau^{X:Y}&=S\left(p^X\right)-p_\tau^Y(1)S\left(p^{X|1}\right)\notag\\
 &=S\left(p^X\right)-(1-\Wass)S\left(\tilde{p}_\Wass^X\right) 
\end{align}
holds, which verifies the optimality of $p_\tau^{XY}$ defined by Eq. (10). 
%This distribution is illustrated in Fig. \ref{fig_S0}.
%\begin{figure}[tbp]
%\centering
%\includegraphics[scale=0.6]{fig_result4.pdf}
%\caption{Illustration of the optimal distribution $p_\tau^{XY}$ and $p^X$. The horizontal axes represent the states, and the vertical axis represents the probability. The thick dashed lines indicate equal probabilities.}
%\label{fig_S0}
%\end{figure}
\subsection{The proof of the speed limits}\label{subsec_A5}
From the Lemma \ref{lem_S1}, the restriction of $I_{p^X}$ to $[0,1-p_\downarrow^X(1)]$ is an injective mapping. Therefore, we can define an inverse function $\Wass_{p^X}:\left[0,S(p^X)\right]\to[0,1-p_\downarrow^X(1)]$ for the fixed $p^X$ as
\begin{equation}
I_{p^X}(\Wass_{p^X}(I))=I.
\end{equation}
The graphical construction of the inverse function is illustrated in Fig. \ref{fig_S1}. Since $I_{p^X}$ is monotonically increasing, $\Wass_{p^X}$ is also monotonically increasing. 
\begin{figure}[tbp]
\centering
\includegraphics[scale=0.6]{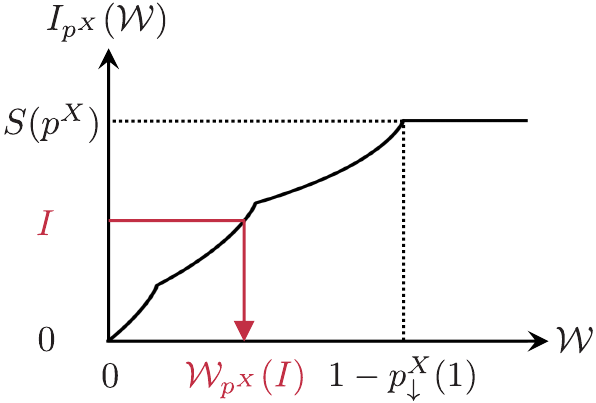}
\caption{Construction of the inverse function$\Wass_{p^X}$ of $I_{p^X}$. For $I\in[0,S\left(p^X\right)$], we define $\Wass_{p^X}(I)\in[0,1-p_\downarrow^X(1)]$ as $I=I_{p^X}\left(\Wass_{p^X}(I)\right)$.}
\label{fig_S1}
\end{figure}
Then, from the Main Theorem, we obtain
\begin{equation}
\Wass_{p^X}(I_\tau^{XY})\leq\Wass_{p^X}\left(I_{p^X}(\Wass)\right)=\Wass,\label{eq_S37}
\end{equation}
where $\Wass=\Wass\left(p_0^{XY},p_\tau^{XY}\right)$. From inequalities (\ref{eq_S37}) and (8) in the main text, we get the speed limits (inequalities (\ref{eq_11}) and (\ref{eq_act})) in the main text.
\subsection{The formalism by Lorenz curve}\label{subsec_A6}
\begin{figure}[tbp]
\centering
\includegraphics[scale=0.5]{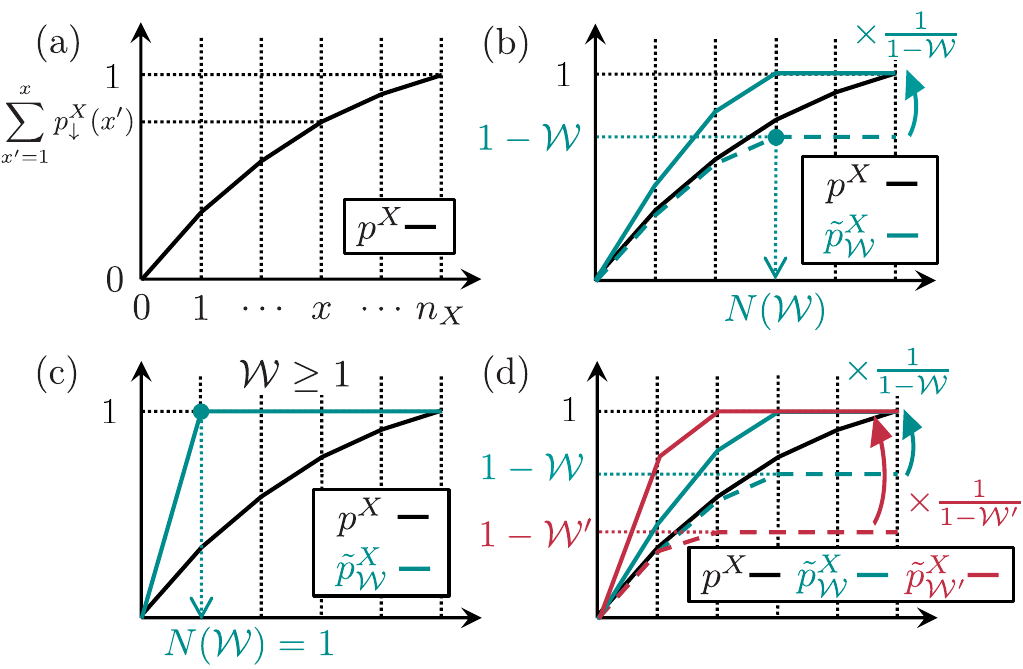}
\caption{(a) The Lorenz curve of $p^X$. (b) Definition of $\tilde{p}_\Wass^X$ and $N(\Wass)$ for $\Wass<1$. First, we obtain a polygonal line (green dashed line) by truncating the Lorenz curve of $p^X$ by the value $1-\Wass$. Then, we define $\tilde{p}_\Wass^X$ as a probability distribution corresponding to the Lorenz curve obtained by scaling the green dashed line by a factor of $1/(1-\Wass)$ (green solid line). $N(\Wass)$ is defined as the minimum value of $x$ which satisfies $\sum_{x'=1}^xp_\downarrow^X(x')>1-\Wass$ (if $\Wass=0$, $N(\Wass)\coloneqq n_X$). (c) Definition of $\tilde{p}_\Wass^X$ and $N(\Wass)$ for $\Wass\geq1$. (d) Comparison of Lorenz curves of $\tilde{p}_\Wass^X$ and $\tilde{p}_{\Wass'}^X$ when $\Wass'>\Wass$. The Lorenz curve of $\tilde{p}_{\Wass'}^X$ lies above that of $\tilde{p}_{\Wass}^X$.}
\label{fig_S1.5}
\end{figure}
We can also visually define $\tilde{p}_\Wass^X$ by the Lorenz curve \cite{Bhatia_2013_MatrixAnalysis,Sagawa_2022_Majorization}. The Lorenz curve is a polygonal line constructed as follows. First, let $p^X_\downarrow$ be the probability distribution obtained by rearranging $p^X$ in the descending order $p_\downarrow^X(1)\geq p_\downarrow^X(2)\geq\cdots\geq p_\downarrow^X(n_X)$. Then, we plot the cumulative probabilities $\sum_{x'=1}^xp_\downarrow^X(x')$ as a function of $x$, which gives the Lorenz curve of $p^X$ (Fig. \ref{fig_2} (a)). The Lorenz curve visualizes how uniformly probabilities are distributed among the states; the entropy $S\left(p^X\right)$ increases as the Lorentz curve of $p^X$ is located lower. \par
Now, using the Lorenz curve, $\tilde{p}_\Wass^X$ can be defined as a probability distribution which corresponds to the green Lorenz curve in Fig. \ref{fig_2} (b), when $\Wass<1$. Here, we define $N(\Wass)\in\mathcal{X}$ as depicted in Fig. \ref{fig_2} (b), which is used when discussing the conditions for achieving the equality in (\ref{eq_9}) . When $\Wass\geq1$, we define $N(\Wass)$ as 1, and $\tilde{p}_\Wass^X$ as a probability distribution corresponding to the green Lorenz cure in Fig. \ref{fig_2} (c).\par
We can also intuitively understand the meaning of inequality (\ref{eq_9}). As shown in Fig. \ref{fig_2} (d), the Lorentz curve of $\tilde{p}_\Wass^X$ moves upward as $\Wass$ increases, resulting in a decrease in $S\left(\tilde{p}_\Wass^X\right)$. Therefore, $I_{p^X}(\Wass)$ is an increasing function of $\Wass$, indicating that a larger Wasserstein distance $\Wass$ allows $Y$ to obtain more mutual information from $X$.
%
%
%
% The property of ratio
%
%
%
\section{The property of the ratio}\label{sec_B}
We here derive the following Proposition as a property of the fundamental bound obtained in this paper.
\setcounter{prop}{1}
\begin{prop}\label{prop_S2}
The ratio 
\begin{equation}
\frac{I_{p^X}(\Wass)}{\Wass^2/(\tau\langle m\rangle_\tau)}\eqqcolon\tau\langle m\rangle_\tau g(\Wass)
\end{equation}
obtained by dividing the upper bound on mutual information $I_{p^X}(\Wass)$ by the minimized entropy production $\min\Sigma_\tau^Y=\Wass^2/(\tau\langle m\rangle_\tau)$ for the fixed time-averaged mobility $\langle m\rangle_\tau$ is a decreasing function of $\Wass$.
\end{prop}
\begin{proof}
We show that the function $g(\Wass)= I_{p^X}(\Wass)/(\Wass^2)$ is decreasing. For $\Wass\geq1-p_\downarrow^X(1)$, $g(\Wass)=S\left(p^X\right)/(\Wass^2)$ is decreasing. Then, we consider the case where $0<\Wass<1-p_\downarrow^X(1)$. From the definition of $\tilde{p}_\Wass^X$, $S\left(\tilde{p}_\Wass^X\right)$ is differentiable function of $\Wass$ for $\Wass\neq\Wass_n\coloneqq\sum_{x=n}^{n_X}p_\downarrow^X(x)\ (n=2,\cdots,\ n_X)$, and the derivative is given by
\begin{align}
&\frac{\diff}{\diff\Wass}S\left(\tilde{p}_\Wass^X\right)\notag\\
&=\frac{\diff}{\diff\Wass}\left[-\sum_{x=1}^{N(\Wass)-1}\frac{p^X(x)}{1-\Wass}\ln\frac{p^X(x)}{1-\Wass}\right.\notag\\
&\hspace{30pt}\left.-\left(1-\sum_{x=1}^{N(\Wass)-1}\frac{p^X(x)}{1-\Wass}\right)\ln\left(1-\sum_{x=1}^{N(\Wass)-1}\frac{p^X(x)}{1-\Wass}\right)\right]\notag\\
&=-\sum_{x=1}^{N(\Wass)-1}\frac{p^X(x)}{(1-\Wass)^2}\left[\ln\frac{p^X(x)}{1-\Wass}+1\right]\notag\\
&\hspace{30pt}+\sum_{x=1}^{N(\Wass)-1}\frac{p^X(x)}{(1-\Wass)^2}\left[\ln\left(1-\sum_{x=1}^{N(\Wass)-1}\frac{p^X(x)}{1-\Wass}\right)+1\right]\notag\\
&=\frac{S\left(\tilde{p}_\Wass^X\right)+\ln{\tilde{p}_\Wass^X(N(\Wass))}}{1-\Wass}.
\end{align}
By using this, we get
\begin{align}
\frac{\diff}{\diff\Wass}g(\Wass)
&=\frac{\diff}{\diff\Wass}\left[\frac{S\left(p^X\right)}{\Wass^2}+\frac{1-\Wass}{\Wass^2}S\left(\tilde{p}_\Wass^X\right)\right]\notag\\
&=-2\frac{S\left(p^X\right)}{\Wass^3}+\left[-\frac{2}{\Wass^3}+\frac{1}{\Wass^2}\right]S\left(\tilde{p}_\Wass^X\right)\notag\\
&\hspace{20pt}+\frac{1}{\Wass^2}\left[S\left(\tilde{p}_\Wass^X\right)+\ln{\tilde{p}_\Wass^X(N(\Wass))}\right]\notag\\
&=-2\frac{S\left(p^X\right)}{\Wass^3}-2\frac{1-\Wass}{\Wass^3}S\left(\tilde{p}_\Wass^X\right)\notag\\
&\hspace{20pt}+\frac{1}{\Wass^2}\ln{\tilde{p}_\Wass^X(N(\Wass))}\notag\\
&\leq0
\end{align}
for $\Wass\neq\Wass_n\ (n=2,\cdots,\ n_X)$. Here, for $n=2,\ 3,\cdots,\ n_X$,
\begin{align}
\lim_{\Wass\nearrow\Wass_n}&S\left(\tilde{p}_\Wass^X\right)\notag\\
&=\lim_{\Wass\nearrow\Wass_n}\left[-\sum_{x=1}^{n-1}\frac{p^X(x)}{1-\Wass}\ln\frac{p^X(x)}{1-\Wass}\right.\notag\\
&\hspace{20pt}\left.-\left(1-\sum_{x=1}^{n-1}\frac{p^X(x)}{1-\Wass}\right)\ln\left(1-\sum_{x=1}^{n-1}\frac{p^X(x)}{1-\Wass}\right)\right]\notag\\
&=-\sum_{x=1}^{n-1}\frac{p^X(x)}{1-\Wass_n}\ln\frac{p^X(x)}{1-\Wass_n}\notag\\
&=-\sum_{x=1}^{n-1}\frac{p^X(x)}{1-\Wass_n}\ln\frac{p^X(x)}{1-\Wass_n}\notag\\
&=S\left(\tilde{p}_{\Wass_n}^X\right)
\end{align}
holds, which implies that $S\left(\tilde{p}_\Wass^X\right)$ is continuous at $\Wass=\Wass_n\ (n=2,\cdots,\ n_X)$. Thus, $g(\Wass)$ is a decreasing function of $\Wass$.
\end{proof}
\begin{prop}\label{prop_S3}
The ratio 
\begin{equation}
\frac{I_{p^X}(\Wass)}{2\Wass\tanh^{-1}\left(\Wass/(\tau\langle a\rangle_\tau)\right)}\eqqcolon h(\Wass)
\end{equation}
obtained by dividing the upper bound on mutual information $I_{p^X}(\Wass)$ by the minimized entropy production $\min\Sigma_\tau^Y=2\Wass\tanh^{-1}\left(\Wass/(\tau\langle m\rangle_\tau)\right)$ for the fixed time-averaged activity $\langle a\rangle_\tau$ is a decreasing function of $\Wass$.
\end{prop}
\begin{proof}
  In the same way as the Proposition \ref{prop_S2}, it suffices to show that the derivative of $h(\Wass)$ is non-positive for $\Wass\neq\Wass_n$. Since $\tanh^{-1}(x)=\ln\frac{1+x}{1-x}$, we get
  \begin{align}
    \frac{\diff}{\diff x}\tanh^{-1}(x)=\frac{2}{(1+x)(1-x)}.
  \end{align}
  Therefore, by setting $\mathcal{A}\coloneqq\tau\langle a\rangle_\tau$, the derivative of $h(\Wass)$ can be calculated as
  \begin{align}
    \frac{\diff h(\Wass)}{\diff\Wass}
    &=\frac{\diff}{\diff\Wass}\left[\frac{S\left(p^X\right)+(1-\Wass)S\left(\tilde{p}_\Wass^X\right)}{2\Wass\tanh^{-1}\frac{\Wass}{\mathcal{A}}}\right]\notag\\
    &=-\frac{\tanh^{-1}\frac{\Wass}{\mathcal{A}}+2\frac{\mathcal{A}\Wass}{(\mathcal{A}+\Wass)(\mathcal{A}-\Wass)}}{2\left[\Wass\tanh^{-1}\frac{\Wass}{\mathcal{A}}\right]^2}S\left(p^X\right)\notag\\
    &\hspace{15pt}-\frac{\mathcal{A}\Wass(1-\Wass)S\left(\tilde{p}_\Wass^X\right)}{(\mathcal{A}+\Wass)(\mathcal{A}-\Wass)\left[\Wass\tanh^{-1}\frac{\Wass}{\mathcal{A}}\right]^2}\notag\\
    &\hspace{20pt}-\frac{(1-\Wass)S\left(\tilde{p}_\Wass^X\right)}{2\Wass^2\tanh^{-1}\frac{\Wass}{\mathcal{A}}}+\frac{\ln\tilde{p}_\Wass^X\left(N(\Wass)\right)}{2\Wass\tanh^{-1}\frac{\Wass}{\mathcal{A}}}\\
    &\leq0,
  \end{align}
  where we used the fact that $\Wass/\mathcal{A}$ is in the domain of $\tanh^{-1}(x)$, which yields $\mathcal{A}>\Wass$.
\end{proof}
Since the upper bound on mutual information $I_{p^X}(\Wass)$ is an increasing function of $\Wass$, Propositions \ref{prop_S2} and \ref{prop_S3} imply that the mutual information obtained per dissipation decreases as $I_\tau^{X:Y}$ increases. This suggests that incomplete information gain (i.e., the case where $I_\tau^{X:Y}$ is small) might make sense in order to increase the average information gain per dissipation.
%
%
%
% The control protocol for coupled quantum dots
%
%
%
%
%
%
%The derivation of the approximately optimal protocol
%
%
%
\section{The approximately optimal protocol for coupled quantum dots}\label{sec_C}
\subsection{The condition for achieving the equality in the inequality (8)}\label{subsec_C1}
In this Subsection, we apply the discussions in Ref. [44] to the setup of this paper. Considering the case where the transitions are induced by the difference in energy levels and there are no non-conservative forces, the stochastic heat can be expressed as $Q_t^{Y}(y,y'|x)=E_t^{XY}(x,y)-E_t^{XY}(x,y')$, where $E_t^{XY}(x,y)$ denotes the instantaneous energy level of the state $(x,y)$ at time $t$. In this case, the entropy production rate $\sigma_t^{XY}\coloneqq\diff\Sigma_t^{XY}/\diff t$ is given by
\begin{align}
\sigma_t^{XY}&=\frac{\diff}{\diff t} S(p_t^{XY})-\beta \frac{\diff}{\diff t} Q_t^{XY}\\
&=\sum_{x,y}\dot{p}_t^{XY}(x,y)\left[-\ln p_t^{XY}(x,y)-\beta E_t^{XY}(x,y)\right]\\
&=\sum_{x,y,y':y\neq y'}\left[A_t(y,y'|x)-A_t(y',y|x)\right]\ln\frac{e^{-\beta E_t^{XY}(x,y)}}{p_t^{XY}(x,y)}\\
&=\sum_{x,y,y':y\neq y'}A_t(y,y'|x)\ln\frac{e^{-\beta E_t^{XY}(x,y)}p_t^{XY}(x,y')}{e^{-\beta E_t^{XY}(x,y')}p_t^{XY}(x,y)}\\
&=\sum_{x,y,y':y\neq y'}A_t(y,y'|x)F_t(y,y'|x)\label{eq_S47}\\
&=\sum_{x,y,y':y>y'}J_t(y,y'|x)F_t(y,y'|x),
\end{align}
where we defined the frequency of transitions $A(y,y'|x)$, thermodynamic force $F(y,y'|x)$, and probability current $J(y,y'|x)$ from the state $(x,y')$ to $(x,y)$ as
\begin{align}
  A_t(y,y'|x)&\coloneqq R_t(y,y'|x)p^{XY}_t(x,y'),\\
  F_t(y,y'|x)&\coloneqq \ln\frac{A_t(y,y'|x)}{A_t(y',y|x)},\\
  J_t(y,y'|x)&\coloneqq A_t(y,y'|x)-A_t(y',y|x).
\end{align}
\par
We here explain the condition for achieving the equality in inequality (8) in the main text. We define the dynamical state mobility as
\begin{equation}
m_t\coloneqq\sum_{x,y,y':y>y'}\frac{J_t(y,y'|x)}{F_t(y,y'|x)},
\end{equation}
which characterizes the sum of the responses of the probability currents against the thermodynamic forces over all transitions, and we denote the time-averaged mobility as $\langle m\rangle_\tau\coloneqq(1/\tau)\int_0^\tau m_t\diff t$. By applying Cauchy-Schwarz inequality, we obtain 
\begin{align}
\sqrt{\Sigma_\tau^{XY}\tau\langle m\rangle_\tau}&\geq\int_0^\tau \sqrt{\sigma_t^{XY}}\sqrt{m_t}\diff t\notag\\
&\geq \int_0^\tau\sum_{x,y,y':y>y'}\abs{J_t(y,y'|x)}\diff t,\label{eq_S51}
\end{align}
where the equalities holds if thermodynamic force $F_t(y,y'|x)$ is constant for all $x,\ y,\ y'\ t$ with nonzero probability currents $J_t(y,y'|x)$. Here, we have
\begin{equation}
\int_0^\tau\sum_{x,y,y':y>y'}\abs{J_t(y,y'|x)}\diff t\geq\Wass\left(p_0^{XY},p_\tau^{XY}\right),\label{eq_S52}
\end{equation}
where the equality is achieved by transporting probabilities according to the optimal transport plan, which is the solution for the minimization problem in the definition of the Wasserstein distance (for more details, see Ref. [44]). Combining inequalities (\ref{eq_S51}) and (\ref{eq_S52}), we obtain inequality (8) in the main text. The condition for achieving the equality in (8) can be summarized as transporting probabilities according to the optimal transport plan under a uniform constant thermodynamic force.\par
For fixed activity, the condition for achieving eq. (8) in the main text is as follows. First, we define local and mean activity at time $t$ as
\begin{align}
  a_t(y,y'|x)&\coloneqq A_t(y,y'|x)+A_t(y',y|x),\\
  a_t&\coloneqq\sum_{x,y,y':y>y'}a_t(y,y'|x).
\end{align}
We denote its time-averaged value as $\langle a\rangle_\tau\coloneqq(1/\tau)\int_0^\tau a_t\diff t$. By defining local entropy production rate as $\sigma_t^{XY}(y,y'|x)\coloneqq J_t(y,y'|x)F_t(y,y'|x)$, we can obtain
\begin{align}
  \sigma_t^{XY}(y,y'|x)=2J_t(y,y'|x)\tanh^{-1}\frac{J_t(y,y'|x)}{a_t(y,y'|x)}.
\end{align}
Since $x\tanh^{-1}(x/y)$ is a convex function for any $x,\ y$ which satisfy $0\leq x<y,\ y>0$ and an increasing function of $x$, we can apply Jensen's inequality and inequality (\ref{eq_S52}) and get
\begin{align}
  \Sigma_\tau\geq2\Wass\left(p_0^{XY},p_\tau^{XY}\right)\tanh^{-1}\frac{\Wass\left(p_0^{XY},p_\tau^{XY}\right)}{\tau\langle a\rangle_\tau},\label{eq_S_act}
\end{align}
which is equivalent to the inequality (8) in main text for $f(x)=2\tanh^{-1}(x)$. The condition for achieving the Jensen's inequality is keeping local activity $a_t(y,y'|x)$ and probability current $J_t(y,y'|x)$ constant for all $x,\ y,\ y'$ and $t$. Therefore the equality in (\ref{eq_S_act}) can be achieved by transporting probabilities under uniform constant activity and probability current along the optimal transport plan from $p_0^{XY}$ to $p_\tau^{XY}$.
\subsection{The protocol for fixed mobility}\label{subsec_C2}
In this Subsection, we explicitly describe how to construct the approximately optimal protocol for coupled quantum dots, and give the derivation of it. The equality in inequality  (\ref{eq_11}) in the main text is achieved by simultaneously achieving equalities in (9) and (8) in the main text. Denoting the probability distribution by the matrix $[p_t^{XY}(x,y)]_{x,y}$, when $0\leq\Wass\leq 1-p_\downarrow^X(1)=p$, the initial and the optimal final distributions are given by
\begin{equation}
p_0^{XY}=\left[
\begin{array}{cc}
p & 1-p \\
0 & 0
\end{array}\right]
,\ p_\tau^{XY}=\left[
\begin{array}{cc}
p-\Wass & 1-p \\
\Wass & 0
\end{array}\right],\label{eq_S53}
\end{equation}
by which the equality in (9) is achieved. Since the second column representing $x=2$ is invariant, the transport between the states $(2,1)$ and $(2,2)$ does not occur in the optimal transport plan, which implies $J_t(2,1|2)=0$ at any time $t\in[0,\tau]$. This can be asymptotically achieved in the limit $\beta U\to\infty$. This condition can be interpreted as setting the interaction between two quantum dots to be sufficiently larger than the thermal fluctuation to make the probability to find the system in the state $(2,2)$ sufficiently small. This statement can be rigorously shown as follows. From $p_0^{XY}(2,2)=0$ and $R_t(2,1|2)=\Gamma^Y\left[1+e^{\beta(\epsilon_t^Y+U-\mu^Y)}\right]^{-1}$, we have
\begin{align}
  p_t^{XY}(2,2)&=\int_0^t\frac{\diff}{\diff t}p_{t'}^{XY}(2,2)\diff t'\notag\\
  &\leq\int_0^tR_{t'}(2,1|2)p_{t'}^{XY}(2,2)\diff t'\notag\\
  &\leq\int_0^t\frac{\Gamma^Y}{1+e^{\beta(\epsilon_{t'}^Y+U-\mu^Y)}}\diff t'\notag\\
  &\xrightarrow[\beta U\to\infty]{}0,
\end{align}
which implies $p_t^{XY}(2,2)\to0$ as $\beta U\to\infty$. Therefore, we get
\begin{align}
&\left|J_t(2,1|2)\right|\notag\\
&\leq R_t(2,1|2)p_t^{XY}(2,1)+R_t(1,2|2)p_t^{XY}(2,1)\notag\\
&=\frac{\Gamma^Yp_t^{XY}(2,1)}{1+e^{\beta(\epsilon_t^Y+U-\mu^Y)}}+\frac{\Gamma^Ye^{\beta(\epsilon_t^Y+U-\mu^Y)}}{1+e^{\beta(\epsilon_t^Y+U-\mu^Y)}}p_t^{XY}(2,2)\notag\\
&\xrightarrow[\beta U\to\infty]{}0.
\end{align}
When $J_t(2,1|2)=0$, the condition for achieving the equality in (8) is transporting the probability from state $(1,1)$ to $(1,2)$ under a constant thermodynamic force $F_t(2,1|1)=F$.\par
We next derive the control protocol $\{\epsilon_t^Y\}$ achieving this condition. We consider the limit $\beta U\to\infty$, where $R_t(2,1|2)=R_t(1,2|2)=0$ and $p_t^{XY}(1,2)=p-p_t^{XY}(1,1)$ holds. From this and Eq. (\ref{eq_12}) in the main text, the equation $F_t(2,1|1)=F$ is transformed into
\begin{equation}
%\ln\frac{p_t^{XY}(0,0)}{e^{\beta(\epsilon_t^Y-\mu^Y)}[p-p_t^{XY}(0,0)]}=F\Longleftrightarrow 
\epsilon_t^Y=\mu^Y-k_\mathrm{B}T\left[F+\ln\frac{p-p_t^{XY}(1,1)}{p_t^{XY}(1,1)}\right].\label{eq_S54}
\end{equation}
By substituting Eq. (\ref{eq_S54}) into Eq. (\ref{eq_12}) in the main text and setting $r=e^F$, we get
\begin{align}
  R_t(2,1|1)&=\frac{\Gamma^Yr[p-p_t^{XY}(1,1)]}{r[p-p_t^{XY}(1,1)]+p_t^{XY}(1,1)},\notag\\
  R_t(1,2|1)&=\frac{\Gamma^Yp_t^{XY}(1,1)}{r[p-p_t^{XY}(1,1)]+p_t^{XY}(1,1)}.\label{eq_S57}
\end{align}
From this and the time evolution (\ref{eq_2}) in the main text, we have 
\begin{equation}
\dot{p}_t^{XY}(1,1)=-\Gamma^Y(r-1)\frac{p_t^{XY}(1,1)\left[p-p_t^{XY}(1,1)\right]}{r[p-p_t^{XY}(1,1)]+p_t^{XY}(1,1)},
\end{equation}
which leads to
\begin{equation}
\quad\frac{\diff}{\diff t}\ln\frac{{p_t^{XY}(1,1)}^r}{p-p_t^{XY}(1,1)}=-\Gamma^Y(r-1).\label{eq_S55}
\end{equation}
By using the function $\phi_{p,r}(x)\coloneqq\ln\left[{x}^r/(p-x)\right]$, the solution for the differential equation (\ref{eq_S55}) can be expressed as
\begin{equation}
\phi_{p,r}\left(p_t^{XY}(1,1)\right)-\phi_{p,r}\left(p_{t_0}^{XY}(1,1)\right)=-\Gamma^Y(r-1)t,\label{eq_S56}
\end{equation}
given an initial value $p_{t_0}^{XY}(1,1)$ at time $t=t_0$. By substituting $p_t^{XY}(1,1)$ determined by Eq. (\ref{eq_S56}) into Eq. (\ref{eq_S54}), we obtain the protocol to transport the probability from the state $(1,1)$ to $(1,2)$ under the constant thermodynamic force $F$.\par
We here note that this protocol is only asymptotically applicable to the initial and final distributions in (\ref{eq_S53}). Since $\phi_{p,r}$ is a decreasing function on $(0,p)$ which satisfies $\lim_{x\to0}\phi_{p,r}(x)=+\infty$ and $\lim_{x\to p}\phi_{p,r}(x)=-\infty$, we cannot set $t_0=0$ (i.e., $p_{t_0}^{XY}(1,1)=p$) in Eq. (\ref{eq_S56}). Instead, we must wait for $p_t^{XY}(1,1)$ to be $p_{t_0}^{XY}(1,1)=p(1-\Delta)$ for some $\Delta\in(0,1)$, then implement the dynamics determined by Eq. (\ref{eq_S56}). The explicit protocol is as follows.
\begin{enumerate}
\renewcommand{\labelenumi}{(\Roman{enumi})}
\item Let $\Delta\in(0,1/(1+r))$. For $0\leq t\leq t_0$, where $t_0$ is defined according to $\Delta$ later, set
\begin{equation}
\epsilon_t^Y=\mu^Y-k_\mathrm{B}T\left[F+\ln\frac{\Delta}{1-\Delta}\right].
\end{equation}
This is obtained by substituting $p_t^{XY}(1,1)=p(1-\Delta)$ into Eq. (\ref{eq_S54}). Then, we have
\begin{equation}
R_t(2,1|1)=\frac{\Gamma^Yr\Delta}{r\Delta+1-\Delta},\quad R_t(1,2|1)=\frac{\Gamma^Y(1-\Delta)}{r\Delta+1-\Delta},
\end{equation}
which yields
\begin{equation}
p_t^{XY}(1,1)=p\exp\left(-\frac{1-\Delta-r\Delta}{1-\Delta+r\Delta}\Gamma^Yt\right).
\end{equation}
Here, $t_0$ is defined by the condition $p_{t_0}^{XY}(1,1)=p(1-\Delta)$, which is transformed into
\begin{equation}
t_0=-\frac{1-\Delta+r\Delta}{1-\Delta-r\Delta}\frac{\ln(1-\Delta)}{\Gamma^Y},
\end{equation}
which vanishes in the limit $\Delta\to 0$. This process is accompanied by the non-optimal entropy production calculated as
\begin{align}
\Sigma_{t_0}^{XY}
&=S\left(p_{t_0}^{XY}\right)-S\left(p_{0}^{XY}\right)-\beta Q_{t_0}^Y\\
&=-p(1-\Delta)\ln[p(1-\Delta)]-p\Delta\ln(p\Delta)\notag\\
&\hspace{20pt}+p\ln p+p\Delta\ln\frac{r\Delta}{1-\Delta}\\
&=-p\ln(1-\Delta)+p\Delta\ln r,
\end{align}
which also vanishes in the limit $\Delta\to 0$.
\item For $t_0\leq t\leq\tau$, set $\epsilon_t^Y$ as Eq. (\ref{eq_S54}), where $p_t^{XY}(1,1)$ satisfies Eq. (\ref{eq_S56}). Since $p_{t_0}^{XY}(1,1)=p(1-\Delta)$ and $p_\tau^{XY}(1,1)=p-\Wass$, the end time $\tau$ is determined by
\begin{equation}
\tau=\frac{\phi_{p,r}(p-\Wass)-\phi_{p,r}(p(1-\Delta))}{\Gamma^Y(r-1)}.
\end{equation}
We can modify $\tau$ by adjusting $\Gamma^Y$. We here note that the protocol for $\Wass=p$ cannot be implemented in this way because $\lim_{\Wass\to }\phi_{p,r}(p-\Wass)=+\infty$. However, we can set $\Wass$ to be arbitrarily close to $p$. 
\end{enumerate}
\subsection{The protocol for fixed activity}\label{subsec_C3}
In this Subsection, we explicitly describe how to construct the approximately optimal protocol for coupled quantum dots in the case where activity is fixed. The equality in inequality  (\ref{eq_12}) in the main text is achieved by simultaneously achieving equalities in (9) and (8) in the main text. Denoting the probability distribution by the matrix $[p_t^{XY}(x,y)]_{x,y}$, when $0\leq\Wass\leq 1-p_\downarrow^X(1)=p$, the initial and the optimal final distributions are given by (\ref{eq_S53}) by which the equality in (9) is achieved.\par
In this case, the condition for achieving the equality in  (\ref{eq_12}) is keeping constant activity $a_t(2,1|1)=a$ and probability current $J_t(2,1|1)=J=\Wass/\tau$. This conditions are transformed as
\begin{align}
  R_t(2,1|1)&=\left(a+\frac{\Wass}{\tau}\right)\frac{1}{2p_t^{XY}(1,1)},\label{eq_A_69}\\
  R_t(1,2|1)&=\left(a-\frac{\Wass}{\tau}\right)\frac{1}{2p_t^{XY}(1,2)}\label{eq_A_70},
\end{align}
from which we can obtain
\begin{align}
  \epsilon_t^Y&=\mu^Y+k_\mathrm{B}T\ln\left[\frac{(a\tau-\Wass)p_t^{XY}(1,1)}{(a\tau+\Wass)p_t^{XY}(1,2)}\right],\label{eq_A_71}\\
  \Gamma_t^Y&=\frac{(a\tau-\Wass)p_t^{XY}(1,1)+(a\tau+\Wass)p_t^{XY}(1,2)}{2p_t^{XY}(1,1)p_t^{XY}(1,2)}a.\label{eq_A_72}
\end{align}
Here, when the initial time is set to $t_0=0$, $p_{t_0}^{XY}(1,2)=0$ is required, which inplies $\epsilon_t^Y\to\infty$ and$\Gamma_t^Y\to\infty$ at $t=t_0$. Therefore, it is necessary to wait until $t=t_0(>0)$ for a positive
probability $\Delta$ to be stored in the state $(1,2)$. the rigorous protocol is as follows:
\begin{enumerate}
\item[(I)] Let $\Delta>0$ be sufficiently small constant. From time$t=0$ to $t=t_0$ which is determined below, set $\epsilon_t^Y$ and $\Gamma_t^Y$ to constant values
\begin{align}
\epsilon_t^Y&=\mu^Y+k_\mathrm{B}T\ln\frac{(a\tau-\Wass)(1-\Delta)}{(a\tau+\Wass)\Delta},\label{eq_A_73}\\
\Gamma_t^Y&=\Gamma^Y=\frac{(a\tau-\Wass)(1-\Delta)+(a\tau+\Wass)\Delta}{2p(1-\Delta)\Delta}a.\label{eq_A_74}
\end{align}
which implies
\begin{equation}
  R_t(2,1|1)=\frac{a+\frac{\Wass}{\tau}}{2p(1-\Delta)},\quad R_t^{Y|1}(1,2)=\frac{a-\frac{\Wass}{\tau}}{2p\Delta}
\end{equation}
and
\begin{equation}
p_t^{XY}(1,1)=pe^{-\Gamma^Yt}+p_\Delta\left(1-e^{-\Gamma^Yt}\right).
\end{equation}
Therefore, by substituting $p_{t_0}^{XY}(1,1)=p(1-\Delta)$, the waiting time $t_0$ is determined by
\begin{equation}
t_0=\frac{1}{\Gamma^Y}\ln\frac{p-p_\Delta}{p(1-\Delta)-p_\Delta},
\end{equation}
which converges to $0$ in the limit $\Delta\to 0$. This waiting process accompanies non-optimal energy cost
\begin{align}
\Sigma_{t_0}^{XY}
&=S\left(p_{t_0}^{XY}\right)-S\left(p_{0}^{XY}\right)-\beta Q_{t_0}^Y\\
&=-p(1-\Delta)\ln[p(1-\Delta)]-p\Delta\ln(p\Delta)\notag\\
&\hspace{15pt}+p\ln p-p\Delta\ln\frac{(a\tau-\Wass)(1-\Delta)}{(a\tau+\Wass)\Delta}\\
&=-p\ln(1-\Delta)+p\Delta\ln\frac{a\tau+\Wass}{a\tau-\Wass},
\end{align}
which also converges to $0$ in the limit $\Delta\to 0$.
\item[(II)] For $t_0\leq t\leq\tau$, let $\epsilon_t^Y$ and $\Gamma_t^Y$satisfy Eqs. (\ref{eq_A_71}) and (\ref{eq_A_72}).
\end{enumerate}
\subsection{Comparison between the optimal and non-optimal protocols 
 for inequality (\ref{eq_11})}\label{subsec_C4}
Here we see an example of a non-optimal protocol, in particular, the protocol that approximately achieves the equality in (\ref{eq_1}) but does not achieve the equality in  (\ref{eq_11}). In other words, although the minimum energy cost for the transformation from $p_0^{XY}$ to $p_\tau^{XY}$ is almost achieved, $p_\tau^{XY}$ is not the optimal distribution for given $I_\tau^{X:Y}$. The probability distribution $p^X$ is set to be $p^X(0)=p\ (\leq1/2)$ and $p^X(1)=1-p$. In this case, the initial distribution $p_0^{XY}$ and a non-optimal distribution $p_\tau^{XY}$ which satisfies $\Wass\left(p_0^{XY},p_\tau^{XY}\right)=\Wass$ are given by
\begin{equation}
p_0^{XY}=\left[
\begin{array}{cc}
p & 1-p \\
0 & 0
\end{array}\right]
,\ p_\tau^{XY}=\left[
\begin{array}{cc}
p & 1-p-\Wass \\
0 & \Wass
\end{array}\right].\label{eq_S63}
\end{equation}
While the transformation determined by Eq. (\ref{eq_S53}) transports the probability for $x=1$, the transformation determined by Eq. (\ref{eq_S63}) transports the probability for $x=2$. However, since the interaction between $X$ and $Y$ is repulsive in this setting, the transformation that transfers the probability from state $(2,1)$ to $(2,2)$ while keeping $p_t^{XY}(1,1)>p_t^{XY}(1,2)$ cannot be implemented. Therefore, we change the initial distribution of memory $Y$ from $p_0^Y(1)=1$ to $p_0^Y(2)=1$. This change is allowed because the main result of this paper does not depend on the choice of the state that $Y$ initially takes with probability 1. In this case, the initial and final distributions are given by
\begin{equation}
p_0^{XY}=\left[
\begin{array}{cc}
0 & 0 \\
p & 1-p
\end{array}\right]
,\ p_\tau^{XY}=\left[
\begin{array}{cc}
0 & \Wass \\
p & 1-p-\Wass 
\end{array}\right].\label{eq_S64}
\end{equation}
Considering the symmetry, the control protocol for this transformation under the constant thermodynamic force $F$ is approximately given by applying the transformation
\begin{align}
\begin{cases}
p\rightarrow 1-p,\\
\epsilon_t^Y\rightarrow2\mu^Y-U-\epsilon_t^Y,\\
x=1\leftrightarrow x=2,\\
y=1\leftrightarrow y=2,
\end{cases}
\end{align}
to the protocol determined by (I) and (II) in the previous Subsection in the limit $\beta U\to\infty$ and $\Delta\to0$. 
\begin{figure*}[tbp]
\centering
\includegraphics[scale=0.8]{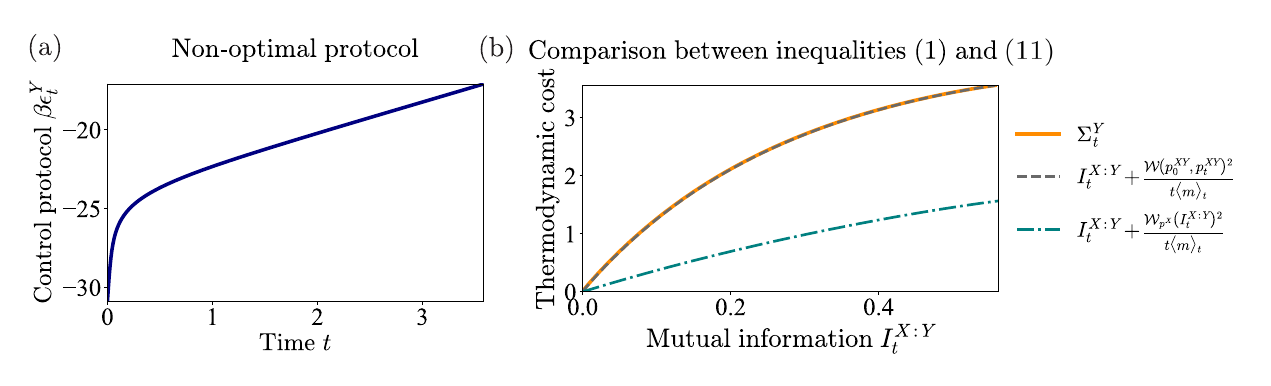}
\caption{Numerical results on the example with two quantum dots. The parameters are set to $\beta U=20$, $F=4$, $\mu^Y=0$, $\Gamma^Y=2$, $\Delta=0.001$. (a) The non-optimal protocol $\epsilon_t^Y$ scaled by $\beta$, which is given by Eq. (S66). (b) Comparison of the left- and right-hand side of inequalities  (\ref{eq_1}) and  (\ref{eq_11}). We see that the protocol shown in (a) approximately (but almost exactly) achieves the equality in (\ref{eq_1}) (gray dashed line), but it does not achieve the equality in  (\ref{eq_11}) (green dash-dot-dash line).}
\label{fig_S2}
\end{figure*}
For $p=0.25$ and $\Wass=0.999p$, the numerical simulation for the non-optimal protocol is shown in Fig. \ref{fig_S2} (a), and the comparison between the lower bound by inequality (\ref{eq_1}) and the lower bound by inequality  (\ref{eq_11}) is shown in Fig. \ref{fig_S2} (b). This shows that the energy cost in this process almost achieves the minimum energy cost determined by inequality (\ref{eq_1}), but does not achieve the minimum energy cost determined by inequality  (\ref{eq_11}). Therefore, this example is not an optimal measurement for obtaining a given amount of mutual information.
\subsection{Comparison between the optimal and non-optimal protocols for inequality (\ref{eq_act})}
Here, we consider a non-optimal protocol that approximately achieves the equality in (\ref{eq_1}) but does not satisfy the equality in (\ref{eq_act}). As well as the previous subsection, we consider the case where the transformation from \(p_0^{XY}\) to \(p_\tau^{XY}\) approximately minimizes the energy cost, while \(p_\tau^{XY}\) itself is not the optimal distribution for given \(I_\tau^{X:Y}\). Assume that \(p^X\) is fixed to \(p^X(0) = p\) (\(p \leq 1/2\)) and \(p^X(1) = 1-p\). The established \(p_0^{XY}\) and the non-optimal \(p_\tau^{XY}\), whose Wasserstein distance from the initial distribution \(p_0^{XY}\) is \(\Wass\), are given by:
\begin{equation}
p_0^{XY} = \left[
\begin{array}{cc}
p & 0 \\
1-p & 0
\end{array}
\right], \quad
p_\tau^{XY} = \left[
\begin{array}{cc}
p & 0 \\
1-p-\Wass & \Wass
\end{array}
\right]\label{eq_S63}
\end{equation}
In contrast to the optimal distribution given by Eq. (\ref{eq_10}), which transfers probability \(\Wass\) from \(y=2\) to \(y=1\), the distribution defined by equation (\ref{eq_S63}) transfers \(\Wass\) from \(y=1\) to \(y=2\). However, due to repulsive interactions between \(X\) and \(Y\) in this setup, a transformation that moves probability from state \(1,2\) to \(2,2\) while keeping the probability of state \(1,1\) is not practically implementable. Therefore, we change the initial distribution of memory $Y$ from $p_0^Y(1)=1$ to $p_0^Y(2)=1$, as done in the previous section. The modified initial and final distributions are given by
\begin{equation}
p_0^{XY} = \left[
\begin{array}{cc}
0 & 0 \\
p & 1-p
\end{array}
\right], \quad
p_\tau^{XY} = \left[
\begin{array}{cc}
0 & \Wass \\
p & 1-p-\Wass
\end{array}
\right]\label{eq_S64}
\end{equation}
The optimal transport between these distributions can be approximately achieved by applying the transformation
\begin{align}
\begin{cases}
&p \rightarrow 1-p,\\
&\epsilon_t^Y \rightarrow 2\mu^Y - U - \epsilon_t^Y,\\
&x=1 \leftrightarrow x=2,\\
&y=1 \leftrightarrow y=2
\end{cases}
\end{align}
to the protocol determined by (I) and (II) in the previous Subsection in the limit $\beta U\to\infty$ and $\Delta\to0$. 
This non-optimal protocol is illustrated in Fig. \ref{fig_nonopt_act} (a) and (b), and the energy cost \(\Sigma_\tau^Y\) is compared with the lower bounds provided by inequalities (\ref{eq_1}) and (\ref{eq_act}) in Fig. \ref{fig_nonopt_act} (c).
\begin{figure*}[tbp]
\centering
\includegraphics[scale=0.8]{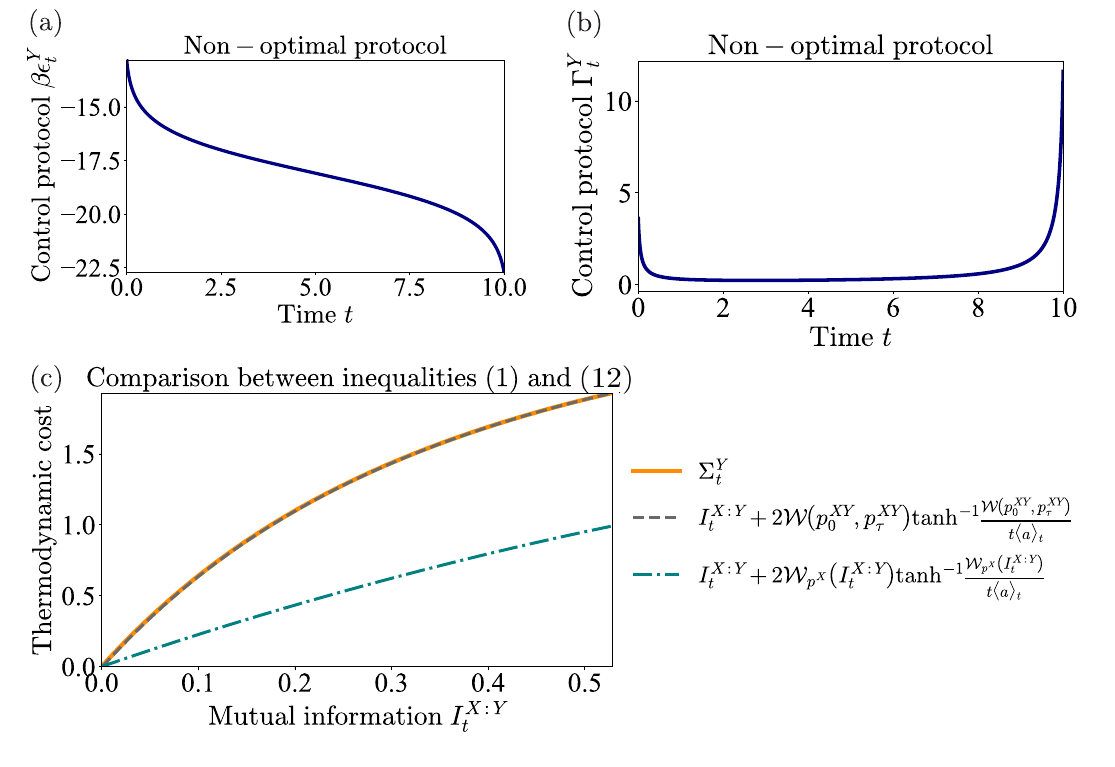}
\caption{Numerical results on the example with two quantum dots. The parameters are set to $\beta U=20$, $A=0.1$, $\tau=10$, $\mu^Y=0$, $\Delta=0.005$. (a) The non-optimal protocol $\epsilon_t^Y$ scaled by $\beta$. (b) (a) The non-optimal protocol $\Gamma_t^Y$. (c) Comparison of the left- and right-hand side of inequalities (\ref{eq_1}) and  (\ref{eq_act}). We see that the protocol shown in (a) approximately (but almost exactly) achieves the equality in (\ref{eq_1}) (gray dashed line), but it does not achieve the equality in  (\ref{eq_12}) (green dash-dot-dash line).}
\label{fig_nonopt_act}
\end{figure*}
\end{appendix}
%\bibliographystyle{apsrev4-2}  %参考文献の書式を指定
%\bibliography{main}　%参考文献を表示
%
%
%参考文献
%
%
\let\oldaddcontentsline\addcontentsline% Store \addcontentsline
\renewcommand{\addcontentsline}[3]{}% Make \addcontentsline a no-op

\end{document}